\documentclass[letterpaper, 10pt, conference]{ieeeconf}      %

\IEEEoverridecommandlockouts                              %

\usepackage{xparse}
\usepackage{float}
\usepackage{amsfonts}
\usepackage{amssymb}
\usepackage{bbm}
\usepackage[Export]{adjustbox} 
\usepackage{nicematrix}
\usepackage{bm}

\usepackage{mathtools}

\usepackage{multirow}
\usepackage{empheq}
\usepackage{todonotes}
\usepackage{accents}
\usepackage{mathrsfs}
\newcommand{\leftrarrows}{\mathrel{\raise.75ex\hbox{\oalign{%
  $\scriptstyle\leftarrow$\cr
  \vrule width0pt height.5ex$\hfil\scriptstyle\relbar$\cr}}}}
\newcommand{\lrightarrows}{\mathrel{\raise.75ex\hbox{\oalign{%
  $\scriptstyle\relbar$\hfil\cr
  $\scriptstyle\vrule width0pt height.5ex\smash\rightarrow$\cr}}}}
\newcommand{\Rrelbar}{\mathrel{\raise.75ex\hbox{\oalign{%
  $\scriptstyle\relbar$\cr
  \vrule width0pt height.5ex$\scriptstyle\relbar$}}}}

\makeatletter
\def\leftrightarrowsfill@{\arrowfill@\leftrarrows\Rrelbar\lrightarrows}
\newcommand{\xleftrightarrows}[2][]{\ext@arrow 3399\leftrightarrowsfill@{#1}{#2}}
\makeatother

\makeatletter
\def\rightleftarrowsfill@{\arrowfill@\lrightarrows\Rrelbar\leftrarrows}
\newcommand{\xrightleftarrows}[2][]{\ext@arrow 3399\rightleftarrowsfill@{#1}{#2}}
\makeatother

\NewDocumentCommand{\tld}{s d()}{
    \IfNoValueTF{#2}{
        ^{\sim}
    }{
        \IfBooleanTF{#1}{
            \tilde{#2}
        }{
            \widetilde{#2}
        }
    }
}

\NewDocumentCommand{\Forall}{s}{
    \IfBooleanTF{#1}{
        \, \forall \,
    }{
        \text{ for all }
    }
}

\newcommand\numberthis{\addtocounter{equation}{1}\tag{\theequation}}

\NewDocumentCommand{\rF}{
    O{R} d()
}{
    \IfNoValueTF{#2}{
        \mathbb{#1}
    }{
        \mathbb{#1}_{#2}
    }
} 
\NewDocumentCommand{\cF}{O{C}}{\mathbb{#1}} %

\NewDocumentCommand{\ii}{s}{
    \IfBooleanTF{#1}{
        \mathbf{i}
    }{
        i
    }
}

\NewDocumentCommand{\tRe}{sd()}{
    \IfBooleanTF{#1}{
        \mathrm{Re}\,
    }{
        \mathrm{Re}\left(#2\right)
    }
}

\NewDocumentCommand{\tIm}{r()}{
    \mathrm{Im}\left(#1\right)
}

\NewDocumentCommand{\sint}{sd()}{
    \IfBooleanTF{#1}{
        
    }{
        \mathrm{int}\,#2
    }
}

\NewDocumentCommand{\mati}{s o}{
    \IfNoValueTF{#2}{
        \IfBooleanTF{#1}{
            \mathbb{I}
        }{
            I
        }
    }{
        \IfBooleanTF{#1}{
            \mathbb{I}_{#2}
        }{
            I_{#2}
        }
    }
}

\NewDocumentCommand{\mato}{s}{
    \IfBooleanTF{#1}{
        \mathbf{0}
    }{
        0
    }
} 

\NewDocumentCommand{\diag}{sm}{
    \IfBooleanTF{#1}{
        \mathrm{diag}
        \left\{
            #2
        \right\}
    }{
        D_{#2}
    }
}

\NewDocumentCommand{\blkdiag}{m}{
    \mathrm{blkdiag}
    \{
        #1
    \}
}

\NewDocumentCommand{\inv}{d()}{
    \IfNoValueTF{#1}{
        ^{-1}
    }{
        \left(#1\right)^{-1}
    }
}

\NewDocumentCommand{\tp}{s d()}{
    \IfBooleanTF{#1}{
        \IfNoValueTF{#2}{
            ^{\mathsf{T}}
        }{
            \left(#2\right)^{\mathsf{T}}
        }
    }{
        \IfNoValueTF{#2}{
            ^\mathsf{T}
        }{
            \left(#2\right)^{\mathsf{T}}
        }
    }
}

\NewDocumentCommand{\ct}{s d()}{
    \IfBooleanTF{#1}{
        \IfNoValueTF{#2}{
            ^{\mathsf{H}} 
        }{
            \left(#2\right)^{\mathsf{H}}
        }
    }{
        \IfNoValueTF{#2}{
            ^{\mathsf{H}} 
        }{
            \left(#2\right)^{\mathsf{H}}
        }
    }
}

\NewDocumentCommand{\cj}{r()}{
    \overline{#1}
}

\NewDocumentCommand{\proj}{s O{}}{
    \IfBooleanTF{#1}{
        \IfNoValueTF{#2}{
            \bm{\Pi}
        }{
            \bm{\Pi}_{#2}
        }
    }{
        \IfNoValueTF{#2}{
            \Pi
        }{
            \Pi_{\mathrm{#2}}
        }
    }
}

\NewDocumentCommand{\norm}{r() o}{
    \IfNoValueTF{#2}{
        \left\|#1\right\|
    }{
        \left\|#1\right\|_{#2}
    }
}

\NewDocumentCommand{\iprod}{r() r() o}{
    \IfNoValueTF{#3}{
        \left \langle #1, #2 \right \rangle
    }{
        \left \langle #1, #2 \right \rangle_{#3}
    } 
}

\usetikzlibrary{calc}

\def\mset{\mathcal{M}}

\def\symset{\mathcal{X}}
\def\symreg{\mathbf{X}}
\def\symsettld{\tilde{\mathcal{X}}}
\def\geqsl{\geqslant}
\def\leqsl{\leqslant}
\def\uncert{\varXi}
\def\gmin{\underline{\gamma}}
\def\gmax{\overline{\gamma}}

\NewDocumentCommand{\infreg}{sd[]}{
    {\IfBooleanTF{#1}{\mathcal{X}}{\mathbf{X}}}^{\mathrm{inf}}\IfNoValueF{#2}{_{#2}}
}
\NewDocumentCommand{\supreg}{sd[]}{
    {\IfBooleanTF{#1}{\mathcal{X}}{\mathbf{X}}}^{\mathrm{sup}}\IfNoValueF{#2}{_{#2}}
}

\newcommand{\svmax}{\overline{\sigma}}
\newcommand{\svmin}{\underline{\sigma}}

\newcommand{\rhinf}{\mathcal{RH}_\infty}
\def\hinf{{\mathcal{H}_\infty}}
\NewDocumentCommand{\nrmhinf}{r()}{%
    \|#1\|_\hinf
}
\NewDocumentCommand{\anghinf}{D[]{\theta}r()}{%
    \Phi_{#1}(#2)_\hinf
}

\NewDocumentCommand{\angmin}{o}{
    \underline{\varphi}\IfNoValueF{#1}{_{#1}}
}
\NewDocumentCommand{\angmax}{o}{
    \overline{\varphi}\IfNoValueF{#1}{_{#1}}
}
\NewDocumentCommand{\ang}{o}{
    \varphi\IfNoValueF{#1}{_{#1}}
}

\NewDocumentCommand{\sphmin}{o}{\underline{\psi}\IfNoValueF{#1}{_{#1}}}
\NewDocumentCommand{\sphmax}{o}{\overline{\psi}\IfNoValueF{#1}{_{#1}}}
\NewDocumentCommand{\sph}{o}{\psi\IfNoValueF{#1}{_{#1}}}

\NewDocumentCommand{\kmin}{D[]{k}}{\underline{#1}}
\NewDocumentCommand{\kmax}{D[]{k}}{\overline{#1}}

\NewDocumentCommand{\degmin}{d()}{
    \underline{d}\IfNoValueF{#1}{_{#1}}
}
\NewDocumentCommand{\degmax}{d()}{
    \overline{d}\IfNoValueF{#1}{_{#1}}
}
\newcommand{\invmap}{f_{\mathrm{inv}}}

\newcommand{\ztld}{\tilde{z}}
\def\hsp{\mathbb{H}}

\NewDocumentCommand{\convhull}{s r()}{
    \IfBooleanTF{#1}{
        \mathrm{co}\,\left\{#2\right\}
    }{
        \mathrm{co}\,#2
    }
}

\NewDocumentCommand{\chordhull}{s D[]{\theta} r()}{
    \IfBooleanTF{#1}{
        \mathrm{ch}^{#2}\,\left\{#3\right\}
    }{
        \mathrm{ch}^{#2}\,#3
    }
}

\NewDocumentCommand{\circhull}{s D[]{\theta} r()}{
    \IfBooleanTF{#1}{
        \mathrm{cir}^{#2}\,\left\{#3\right\}
    }{
        \mathrm{cir}^{#2}\,#3
    }
}

\NewDocumentCommand{\parahull}{s r()}{
    \IfBooleanTF{#1}{
        \mathrm{para}\,\left\{#2\right\}
    }{
        \mathrm{para}\,#2
    }
}

\NewDocumentCommand{\conihull}{s r()}{
    \IfBooleanTF{#1}{
        \mathrm{cone}\,\left\{#2\right\}
    }{
        \mathrm{cone}\,#2
    }
}
\NewDocumentCommand{\starhull}{s r()}{
    \IfBooleanTF{#1}{
        \mathrm{star}\,\left\{#2\right\}
    }{
        \mathrm{star}\,#2
    }
}

\NewDocumentCommand{\gasys}{s}{
    \IfBooleanTF{#1}{\mathcal{D}_{\circ}}{\mathcal{D}}
}

\NewDocumentCommand{\cone}{s}{
    \IfBooleanTF{#1}{\mathbf{C}}{\mathbf{C}_{\circ}}
}

\NewDocumentCommand{\cosys}{s}{
    \IfBooleanTF{#1}{\mathcal{C}_{\circ}}{\mathcal{C}}
}

\NewDocumentCommand{\sesys}{s}{
    \IfBooleanTF{#1}{\mathcal{S}_{\circ}}{\mathcal{S}}
}

\NewDocumentCommand{\epi}{s r()}{
    \IfBooleanTF{#1}{
        \mathbf{epi}\,\left\{#2\right\}
    }{
        \mathbf{epi}\,#2
    }
}

\NewDocumentCommand{\scomp}{s d()}{
    \IfBooleanTF{#1}{
        \IfNoValueTF{#2}{
            ^{\mathsf{c}} 
        }{
            \left(#2\right)^{\mathsf{c}}
        }
    }{
        \IfNoValueTF{#2}{
            ^{\mathsf{c}} 
        }{
            \left(#2\right)^{\mathsf{c}}
        }
    }
}

\NewDocumentCommand{\dwshell}{sD[]{\mathbf{DW}}r()}{
    \IfBooleanTF{#1}{\widetilde{#2}}{#2}(#3)
}

\NewDocumentCommand{\invdwshell}{sD[]{\mathbf{DW}}r()}{
    \IfBooleanTF{#1}{\widetilde{#2}^{-1}}{#2^{-1}}(#3)
}

\NewDocumentCommand{\srg}{sd[]r()}{
    \mathbf{SRG}\IfNoValueF{#2}{_{\IfBooleanTF{#1}{\mathrm{#2}}{#2}}}(#3)
}

\NewDocumentCommand{\invsrg}{d[]r()}{
    \mathbf{SRG}\IfNoValueF{^{-1}}{^{-1}_{\mathrm{#1}}}(#2)
}

\NewDocumentCommand{\disc}{s D[]{} d()}{
    \IfBooleanTF{#1}{
        \mathbf{D}^{\mathsf{c}}
    }{
        \mathbf{D}
    }
    \IfNoValueTF{#2}{}{_{#2}}
    \IfNoValueTF{#3}{}{
        (#3)
    }
} 

\NewDocumentCommand{\region}{sD[]{R}d()}{
    \IfBooleanTF{#1}{
        \mathbf{#2}^{\mathsf{c}}
    }{
        \mathbf{#2}
    }
    \IfNoValueTF{#3}{}{
        (#3)
    }
}

\NewDocumentCommand{\sector}{sD[]{}d()}{
    \IfBooleanTF{#1}{
        \mathbf{S}^{\mathsf{c}}
    }{
        \mathbf{S}
    }
    \IfNoValueTF{#2}{}{_{#2}}
    \IfNoValueTF{#3}{}{
        [#3]
    }
}

\NewDocumentCommand{\lineseg}{sD[]{}d()}{
    \IfBooleanTF{#1}{
        \mathbf{L}^{\mathsf{c}}
    }{
        \mathbf{L}
    }
    \IfNoValueTF{#2}{}{_{#2}}
    \IfNoValueTF{#3}{}{
        (#3)
    }
}

\NewDocumentCommand{\parab}{sD[]{}d()}{
    \IfBooleanTF{#1}{
        \mathbf{P}^{\mathsf{c}}
    }{
        \mathbf{P}
    }
    \IfNoValueTF{#2}{}{_{#2}}
    \IfNoValueTF{#3}{}{
        [#3]
    }
}

\NewDocumentCommand{\seg}{s}{
    \IfBooleanTF{#1}{
        \mathbf{Seg}^{\mathsf{c}}
    }{
        \mathbf{Seg}
    }
}

\NewDocumentCommand{\plane}{sD[]{}d()}{
    \IfBooleanTF{#1}{
        \mathbf{H}^{\mathsf{c}}
    }{
        \mathbf{H}
    }
    \IfNoValueTF{#2}{}{_{#2}}
    \IfNoValueTF{#3}{}{
        (#3)
    }
}

\NewDocumentCommand{\ssreal}{D(){A}D(){B}D(){C}D(){D}}{%
    \begin{bNiceArray}{@{}c|c@{}}[left-margin=.5em,right-margin=.5em]
        #1 & #2  \\\hline 
        #3 & #4
    \end{bNiceArray}
}

\NewDocumentCommand{\lmi}{D(){X,Y}}{%
    \mathrm{LMI}(#1)
}
\NewDocumentCommand{\mi}{D(){X,Y}}{%
    \mathrm{MI}(#1)
}
\NewDocumentCommand{\milhs}{D(){X,Y}}{%
    H(#1)
}

\def\mdisc{\Theta_{\mathsf{d}}}

\usepackage[norelsize,ruled,lined,linesnumbered]{algorithm2e}
\IncMargin{-.6em}
\usepackage{etoolbox}
\makeatletter
\patchcmd{\@algocf@start}%
    {-1.5em}%
    {0pt}%
    {}{}%
\makeatother

\NewDocumentCommand{\larg}{D[]{c}}{%
    \overline{#1}
}
\NewDocumentCommand{\smal}{D[]{c}}{%
    \underline{#1}
}

\usepackage[capitalise]{cleveref}
\usepackage{subcaption}
\usepackage{adjustbox}

\newtheorem{theorem}{Theorem}
\newtheorem{lemma}{Lemma}
\newtheorem{prop}{Proposition}
\newtheorem{obs}{Observation}
\newtheorem{corol}{Corollary}
\newtheorem{example}{Example}
\newtheorem{defn}{Definition}
\newtheorem{remark}{Remark}
\crefname{defn}{Def.}{Defs.}
\Crefname{defn}{Definition}{Definitions}

\crefname{theorem}{Thm.}{Thms.}
\Crefname{theorem}{Theorem}{Theorems}

\crefname{lemma}{Lem.}{Lems.}
\Crefname{lemma}{Lemma}{Lemmas}

\crefname{prop}{Prop.}{Props.}
\Crefname{prop}{Proposition}{Propositions}

\crefname{remark}{Rem.}{Rems.}
\Crefname{remark}{Remark}{Remarks}

\crefname{obs}{Obs.}{Obs.}
\Crefname{obs}{Observation}{Observations}

\crefname{corol}{Corol.}{Corols.}
\Crefname{corol}{Corollary}{Corollaries}

\makeatletter
\newenvironment{myproof}[1][Proof]{%
  \noindent\hspace{2em}{\itshape #1.\ }%
}{%
  \hspace*{\fill}~\QED\par
}
\makeatother

\title{\LARGE \bf
    Symmetry Is Almost All You Need: Robust Stability with Uncertainty Induced by Symmetric SRG Regions
}

\author{Ding Zhang$^{1}$,  Di Zhao$^{2}$, Philipp Braun$^{1}$, and Jianqi Chen$^{2}$%
\thanks{$^{1}$Ding Zhang and Philipp Braun are with the School of Engineering, The Australian National University, Canberra, Australia {\tt\small ding.zhang@connect.ust.hk, philipp.braun@anu.edu.au}}%
\thanks{$^{2}$Di Zhao and Jianqi Chen are with the School of Robotics and Automation, Nanjing University, Suzhou, Jiangsu, China {\tt\small dizhao@nju.edu.cn, jqchen@nju.edu.cn}}
}

\begin{document}

\maketitle
\thispagestyle{empty}
\pagestyle{empty}

\begin{abstract}

This paper investigates the robust stability problem of a feedback system in the presence of uncertainties induced by graphical regions in the plane where the scaled relative graphs (SRGs) reside.
Our main results are developed using a novel and intuitive concept, the Davis-Wielandt shell, together with its connection to SRGs and related variants.
We first study a matrix robust nonsingularity (MRN) problem for two types of graphically induced uncertainty sets: one with prior information on $\theta$ and one without. 
    In the former case, we show that, whenever the uncertainty-inducing region is mirror symmetric about the $\theta$-axis, the separation between a specific variant of the SRG and the region provides a necessary and sufficient condition for MRN. When the region is asymmetric, the necessity generally fails. This recovers the necessity of the small gain condition, and reveals the necessity of small angle conditions and sectored-disc conditions at the matrix level. 
    In the latter case, we show that an additional $\theta$-circular connectivity property is required to obtain necessary and sufficient conditions. 
Building on these MRN results, we then derive sufficient conditions for robust stability of multi-input multi-output (MIMO) linear time-invariant (LTI) systems under frequencywise symmetric uncertainties. In addition, connections with existing system characteristics such as disc-boundedness are discussed and exploited to obtain state-space characterisations for angle-bounded and mixed gain-angle-bounded systems. Based on these results, we construct a $\theta$-angle-gain profile of a system that provides an intuitive visualisation of its feedback robustness against conic and sectorial uncertainties.

\end{abstract}

\section{INTRODUCTION}
Robust control represents a worst-case analysis and design paradigm and therefore inherently admits conservatism in exchange for guaranteed reliability. Conservatism may enter this paradigm through several channels, among which two prominent ones are the over-estimation of the uncertainty due to \emph{specific system characterisations} and the \emph{innate conservatism of associated analytical tools}. The more accurately the uncertainty is characterised and the less conservative the analytical conditions are, the greater the freedom one would have to design feedback controllers that achieve additional objectives while maintaining guaranteed robustness properties with respect to ground-truth uncertainty.

In contrast to single-input single-output (SISO) linear systems, for which the Nyquist diagram serves as non-conservative system characterisation handy for robustness analysis \cite{vinnicombeUncertaintyFeedbackInfinity2001,qiuIntroductionFeedbackControl2009}, multivariable systems embody  
higher-dimensional information and hence present many possibilities of how they can be described and analysed. Theories such as dissipativity~\cite{willemsDissipativeDynamicalSystems1972a} and Integral Quadratic Constraint (IQC)~\cite{megretskiSystemAnalysisIntegral1997} provide rather general frameworks that encompass a broad class of system characterisations and enable systematic analysis and synthesis procedures. 
Still, these are abstract existence
results, and how to construct a suitable characterisation, e.g., a multiplier, for a given system represents a largely open problem.

From bottom up, fruitful results are obtained around uni-dimensional characteristics such as the classic $\hinf$ gain \cite{zhouRobustOptimalControl1995,vinnicombeUncertaintyFeedbackInfinity2001}, positive real (PR) \cite{andersonSystemTheoryCriterion1967} and negative imaginary (NI) \cite{lanzonStabilityRobustnessFeedback2008} properties.
These have been further enriched by recent advances in uni-dimensional phase-related studies \cite{laibPhaseIQCHierarchical2015,chenPhaseTheoryMultiinput2024,chenSingularAngleNonlinear2025,chenCyclicSmallPhase2025},
    as well as by developments in mixed type characterisations, including
    sector-bounded%
    \footnote{This refers to sector-boundedness defined in terms of time-domain input-output pairs, as in sector-bounded nonlinearities. This notion corresponds to disc-boundedness in the frequency domain and differs from the frequency-domain sector-boundedness considered here.}
    LTI systems \cite{guptaRobustStabilityAnalysis1996,bridgemanConicsectorbasedControlCircumvent2014},
    mixed small-gain-PR systems \cite{griggsMixedSmallGain2007},
    scaled relative graphs (SRGs) \cite{patesScaledRelativeGraph2021,chaffeyGraphicalNonlinearSystem2023,chenSoftHardScaled2025,chenGraphicalDominanceAnalysis2025,baron-pradaMixedSmallGain2025}, and Davis-Wielandt (DW) shells \cite{lestasLargeScaleHeterogeneous2012,zhaoWhenSmallGain2022,zhangLocalStabilityCongestion2025,liangFeedbackStabilityMixed2025,zhangPhantomDavisWielandtShell2025a}.
In particular, the SRG, a recently popularised graphical system characterisation, is a two-dimensional set in the complex plane (referred to as the \emph{SRG plane} hereafter) with favourable algebraic properties \cite{ryuScaledRelativeGraphs2022,chaffeyGraphicalNonlinearSystem2023,patesScaledRelativeGraph2021}, integrating both the system gain and singular angle \cite{chenSingularAngleNonlinear2025,baron-pradaMixedSmallGain2025}. 
Lately, several extensions of the original SRG \cite{eijndenPhaseScaledGraphs2025,zhangPhantomDavisWielandtShell2025a} have been proposed to alleviate its conservatism. Among them, its $\theta$-variant \cite{zhangPhantomDavisWielandtShell2025a} integrates the system gain with %
more general `$\gamma$-segmental phases' \cite{chenCyclicSmallPhase2025}.

\begin{figure}[!t]
    \centering
    \adjincludegraphics[Clip = {.1\width} {0.01\height} {0} {0}, height=2.2cm]{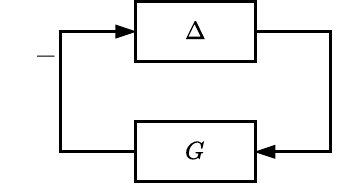}
    \caption{System diagram for the robust stability analysis.\label{fig:sys-diagram}}
    \vspace{-.5cm}
\end{figure}
The aforementioned system characterisations, whether uni-dimensional or combined, all feature graphical interpretations and each can be used to describe an uncertainty set $\Delta$. 
The robust stability (RS) problem then concerns the closed-loop stability of the feedback interconnection shown in \cref{fig:sys-diagram} with respect to such uncertainty sets. 
    A variety of sufficient conditions have been proposed to guarantee robust stability, and their necessity, which often reflects the inherent conservatism attached to the corresponding characterisation, are also of persisting interest (cf. \cite{yangSmallPhaseCondition2025,ringhGainPhaseType2025}).

In terms of the $\theta$-SRG characterisation, systems whose $\theta$-SRG falls within a disc centered at the origin constitute a gain-bounded uncertainty set. 
The associated RS condition is essentially a small-gain condition, which is known to be both necessary and sufficient \cite[Thm.~9.1]{zhouRobustOptimalControl1995}. The robustness with respect to this class of uncertainties can be quantified numerically through the bounded real lemma \cite[Cor.~13.24]{zhouRobustOptimalControl1995}. In contrast, the cases of cone- and sector-boundedness (i.e., angle- and mixed gain-angle-bounded uncertainties) or boundedness by more general regions remain underexplored. 
This paper partially address the RS problem with respect to these SRG region-induced uncertainties. Specifically, we make the following contributions:
    \begin{enumerate}
        \item We propose two ways to induce an uncertainty set based on a prescribed region in the SRG plane, one with a prior specification on $\theta$ and one without.
        We then study a %
        matrix robust nonsingularity problem concerning these uncertainty sets, and reveals the fundamental geometric properties on the inducing regions in both cases so that the SRG-based separation provides a necessary and sufficient condition for this problem. 
        \item Based on the matrix nonsingularity results, we propose sufficient conditions that guarantee robust stability with respect to a broad class of graphically induced uncertainty sets.
        \item We investigate the relationship between cone- and sector-boundedness and disc-boundedness (not necessarily centred at the origin). Based on these relationships, %
        we derive state-space characterisations and construct a graphical $\theta$-angle-gain profile that captures the system's robustness with respect to the conic and sectorial uncertainties.
    \end{enumerate}

\emph{Outline}. Sec.~\ref{sec:problem_formulation} introduces the uncertainty definitions and formulates the matrix robust nonsingularity (MRN) and robust stability (RS) problems. The main results are presented in Secs.~\ref{sec:matrix-ana}-\ref{sec:ss}, which address the MRN (Sec.~\ref{sec:matrix-ana}) and RS (Sec.~\ref{sec:system-ana}) problems, and develop numerical characterisations and a robustness profile (Sec.~\ref{sec:ss}), respectively. 

\subsection{Notation}
Let $\ii* = \sqrt{-1}$ be the imaginary unit. We denote by $\mathbb{R}$, $\mathbb{C}$ the fields of real and complex numbers, respectively.
The sets of positive and nonnegative reals are denoted by $\mathbb{R}_+$ and $\overline{\mathbb{R}}_+$.
A complex number $z\in \cF$ can be represented in Cartesian form $\tRe(z)+\ii*\,\tIm(z)$ and also in polar form $|z|e^{\ii*\angle z}$ where $\angle z$ denotes its argument.
A $\theta$-axis ($\theta \in \rF$) refers to the line $e^{\ii*\theta}\rF$, with the ray in the direction $\theta$ taken as the positive $\theta$-axis and the ray in the direction $\theta+\pi$ as the negative $\theta$-axis. For any $z\in \cF$, we denote by $z_\theta$ its $\theta$-conjugate: $z_\theta := e^{\ii* \theta} \cj(e^{-\ii*\theta} z)$. Note that $z, z_\theta$ are symmetric about the $\theta$-axis.

Given a matrix $A \in \cF^{n\times n}$, we use $A\tp, \cj(A)$, $A\ct*$ to denote its transpose, conjugate, conjugate transpose, respectively. 
$\blkdiag{\cdot}$ denotes a block-diagonal matrix formed by placing the matrices within the braces along the diagonal.
The norm $\|\cdot\|_2$ refers to the induced 2-norm of matrices, which equals the largest singular value of the matrix. We write $H \succcurlyeq \mato$ if $H$ is Hermitian positive semidefinite.

A mapping applied to a set means the image of the set under pointwise evaluation by the mapping.
For two sets $\region[X]$ and $\region[Y]$, $\region[X]\setminus\region[Y]$ denotes the set difference $\{x \in \region[X]: x \notin \region[Y]\}$. 
For a set $\region[S]$ in an ambient space $\hsp$, $\sint(\region[S])$ denotes its interior and $\scomp(\region[S])$ denotes its complement $\hsp\setminus \region[S]$. The epigraph of set $\region[X] \subseteq \hsp\times \rF$, denoted by $\epi(\region[X])$, refers to the set $\{(z,\nu)\in \hsp\times \rF: \nu\geqsl h, (z,h)\in \region[X]\}$.
Typical geometric objects and operations used throughout this paper are defined as in \cref{tab:geo-def}.
In the complex plane, note that the convex hull can be reformulated as a union of $\theta$-chordal hulls:
    $\cup_{\theta\in[0,\pi]} \chordhull[\theta](\region)$.

\begin{table}[t]
\centering
\caption{Definitions of geometric objects and operations.\label{tab:geo-def}}
\begin{NiceTabular}{|@{\,}l@{\hspace{.4em}}p{2.3cm}@{\hspace{.3em}}p{17.2em}@{\,}|}
\hline 
notation & object / operation & definition \\ \hline\hline
$\lambda \region$ & scalar multiplication & $\{\lambda x: x\in \region\}$ \\
$\starhull(\region)$ & star hull & $\cup_{\lambda \in [0,1]} \lambda \region$ \\
$\convhull(\region)$ & convex hull & $\{\lambda z + (1-\lambda)s: z, s \in \region,\lambda \in [0,1]\}$\\
$\chordhull[\theta](\region)$ & $\theta$-chordal hull &  
    $\left\{\lambda z + (1-\lambda) z_\theta : z, z_\theta \in \region, \lambda \in [0,1] \right\}$ %
\\
$\parahull(\region)$ & paraboloidal hull & $\{(\lambda z, \lambda^2 \nu)\in \hsp\times\rF: (z,\nu) \in \region, \lambda \in [0,\infty)\}$\\\hline
$\disc[\gamma]$ & disc (centered at $0$) & $\{z\in \cF: |z|\leqsl \gamma\}$ \\
$\disc[\gamma](s)$ & disc (centered at $s$) & $\{z\in\cF: |z-s|\leqsl \gamma\}$\\
$\cone*[\alpha,\beta]$ & cone & $\{z\in \cF: \alpha\leqsl \angle z \leqsl \beta\}\cup\{0\}$ \\
$\sector[\gamma][\alpha,\beta]$ & sector & $\disc[\gamma]\cap \cone*[\alpha,\beta]$ \\
$\seg[\alpha,\beta]$ & unit disc segment & 
    $\chordhull[\frac{\alpha+\beta}{2}](\{e^{\ii* \phi}: \alpha\leqsl \phi\leqsl \beta\})$\\
$\plane[\gamma]$ & horizontal hyperplane & $\{(z,\nu)\in\cF\times \rF: \nu = \gamma\}$ \\
$\parab[a]$ & paraboloid  & $\{(z,\nu)\in\cF\times \rF: |z|^2 = a^2 \nu\} $ \\\hline
\end{NiceTabular}
\vspace{-.6cm}
\end{table}  

\section{PROBLEM FORMULATION} \label{sec:problem_formulation}
The initial problem studied herein is inspired by \cite{liangFeedbackStabilityMixed2025} and will be formulated in terms of very specific uncertainty sets, namely the disc, conic, and sectorial uncertainties induced by their corresponding regions as defined in \cref{tab:geo-def}, but will be extended to symmetric uncertainties as we develop the main results (in Section \ref{subsec:mrs-generalise}).

\subsection{Disc, Conic, and Sectorial Uncertainties \label{subsec:uncert}}
Let $\hsp$ be a Hilbert space endowed with the inner product $\iprod(\cdot)(\cdot)$ which is linear in the second slot, and the norm $\|\cdot\|$ be the one induced by $\iprod(\cdot)(\cdot)$. 
For a bounded linear operator $G:\hsp\rightarrow \hsp$
and an input $u\in\hsp$,
    the input-wise gain of $G$ at a nonzero $u$ is $\gamma(G;u):= \frac{\|Gu\|}{\|u\|}$,  whereas 
    the input-wise angle%
    \footnote{For the nonzero $u$ with $Gu = 0$, the angle at $u$ is left undefined. However, the gain at $u$ is zero. Thus, the kernel of $G$ ends up being mapped to the origin in the SRG plane (as well as the origin in the Davis-Wielandt %
    shell space).}  
    is $\ang(G;u):= \arccos \frac{\tRe*\iprod(u)(Gu)}{\|Gu\|\|u\|}$.
Then by combining these two parts into a point in the complex plane via the polar form and taking the union of such points over all nonzero $u$, one can obtain the scaled relative graph (SRG) of $G$: 
\begin{align*}
    \srg(G) := 
    \left\{ \gamma(G;u) e^{\pm\ii* \ang(G;u)} : u \neq 0\right\},
\end{align*}
which is an informative two-dimensional graphical representation with rather nice algebraic properties \cite{hannahScaledRelativeGraph2016,ryuScaledRelativeGraphs2022,chaffeyGraphicalNonlinearSystem2023,patesScaledRelativeGraph2021}. 
In addition, given a $\theta \in \rF$,
    the set $\srg[\theta](G)$ associated with $G$ refers to an angular extension of $\srg(G)$, read $\theta$-SRG 
    and can be obtained as $\srg[\theta](G) := e^{\ii*\theta}\srg(e^{-\ii*\theta}G)$.

Let $\mset$ be a prescribed set of linear operators.
Given a gain bound $\gamma >0$ and $\alpha, \beta \in \rF$ with $\alpha\leqsl \beta$, the \emph{disc} uncertainty set and the conic uncertainty set are defined as:
\begin{align} 
    \gasys_\gamma &:=  \left\{\Delta \in \mset: \srg(\Delta) \subseteq \disc[\gamma] \right\}. 
    \label{def:disc} \\
    \cosys[\alpha,\beta] &:= \left\{\Delta \in \mset: \srg[\frac{\beta+\alpha}{2}](\Delta) \subseteq \cone*[\alpha,\beta]\right\},
    \label{def:conic} 
\end{align}respectively. With (\ref{def:disc}) and (\ref{def:conic}), the \emph{sectorial} uncertainty set is defined as the intersection between $\gasys_\gamma$ and $\cosys[\alpha,\beta]$, or more explicitly:
\begin{align}
    \sesys_\gamma[\alpha,\beta] :=&  \left\{\Delta \in \mset : \srg[\frac{\beta+\alpha}{2}](\Delta) \subseteq \sector[\gamma][\alpha,\beta]\right\}. 
    \label{def:sect}
\end{align}
Note that the terms ``cone'' and ``sector'' are often used interchangeably. In this paper, we, however, stringently stipulate that a sector must have a finite gain bound while a cone must not. 
For sets $\cosys[\alpha,\beta]$ and  $\sesys_\gamma[\alpha,\beta]$ that incorporate prescribed angular bounds, 
    they reduce to the set of scalar operators when $\alpha = \beta$. 
    At the other extreme where $\beta-\alpha \geqsl 2\pi$, $\cosys[\alpha,\beta]$ reduces to $\mset$ and $\sesys_\gamma[\alpha,\beta]$, as the intersection between $\gasys_\gamma$ and $\cosys[\alpha,\beta]$, hence coincides with $\gasys_\gamma$.

In the case of stable linear time-invariant systems though, the way
we think about their SRGs will slightly deviate from the unified operator-theoretic perspective that treats them as operators from $\mathcal{L}_2$ to $\mathcal{L}_2$ (a.k.a. the soft SRGs in \cite{chaffeyGraphicalNonlinearSystem2023,patesScaledRelativeGraph2021,chenSoftHardScaled2025,eijndenPhaseScaledGraphs2025,grootExploitingStructureMIMO2025}, for example). Instead, we examine their SRGs frequencywise via their frequency response , namely %
$\srg(\Delta(\ii*\omega))$, where $\Delta(\ii*\omega) \in \cF^{n\times n}$ is treated as a linear transformation on $\cF^n$ (as in \cite{chenGraphicalDominanceAnalysis2025,baron-pradaStabilityResultsMIMO2025}). Then the corresponding $\theta$-SRGs and the uncertainty sets 
are to
be interpreted in line with this %
convention, e.g., $\gasys_\gamma$ refers to $\{\Delta \in \rhinf^{n\times n}:\srg(\Delta(\ii*\omega))\subseteq \disc[\gamma]\,\forall\,\omega\in [0,\infty]\}$.

\subsection{Robust Stability Problem\label{subs:rs-formulation}}
Let $\rhinf^{n\times n}$ be the set of $n$-by-$n$ real rational stable transfer matrices.
We consider the \emph{negative} feedback interconnection as shown in \cref{fig:sys-diagram}, where the loop component $G \in \rhinf^{n\times n}$ is known and $\Delta$ is an uncertain component which is of compatible size and belongs to a prescribed uncertainty set within $\rhinf^{n\times n}$. 
For a fixed $\Delta$, the closed-loop system is stable if $\mati + G\Delta$ has an inverse in $\rhinf^{n\times n}$. 
Then by robust stability, we mean the following:
\begin{defn}[Robust Stability (RS)] \label{def:rs}
    Consider the negative feedback interconnection in Fig. \ref{fig:sys-diagram} with transfer matrix $G \in \rhinf^{n\times n}$. The feedback interconnection defined through $G$ is called robustly stable (RS) with respect to $\uncert \subseteq \rhinf^{n\times n}$ 
    if 
   $\mati + G\Delta$ has an inverse in $\rhinf^{n\times n}$ for every $\Delta \in \uncert$.
\end{defn}

In this paper, we are interested in the following question:
\emph{Which conditions must $G$ satisfy, so that the %
feedback interconnection in Fig. \ref{fig:sys-diagram} is RS with respect to $\uncert$ when $\uncert$ is defined through $\gasys_\gamma$, $\cosys[\alpha,\beta]$, $\sesys_\gamma[\alpha,\beta]$ in \eqref{def:disc}-\eqref{def:sect}, or more general sets?}

\subsection{Matrix Robust Nonsingularity Problem \label{subsec:mrnp}}
    The RS problem, under the specific setup in \cref{subs:rs-formulation} and in the case where $\uncert$ is invariant under contractive scaling, i.e., $\tau \uncert \subseteq \uncert$ for all $\tau \in [0,1]$, boils down to a matrix nonsingularity problem of the form $\mati + G\Delta$. Therefore, we will first examine a robust matrix nonsingularity problem, for which the uncertainty sets as in \eqref{def:disc}-\eqref{def:sect} are instantiated with $\mset = \cF^{n\times n}$ and $\hsp = \cF^n$. Formally, by robust matrix nonsingular property, we mean:
    \begin{defn}[Matrix Robust Nonsingularity (MRN)] \label{defn:mrn}
        A matrix $G \in \cF^{n\times n}$ is said to satisfy the matrix robust nonsingularity (MRN) property with respect to a matrix set $\uncert\subseteq \cF^{n\times n}$, or simply that MRN holds for $G$ w.r.t. $\uncert$, if $\det(\mati+G\Delta)\neq 0$ for all $\Delta \in \uncert$.
    \end{defn} 
    
    Similarly, we are interested in the following question:
   \emph{What conditions must a matrix $G$ satisfy in order to exhibit the MRN property with respect to a matrix uncertainty set $\uncert$ induced by a graphical region in the SRG plane?}

\section{MATRIX ROBUST NONSINGULARITY\label{sec:matrix-ana}}
In this section, we address the MRN problem as set out in \cref{subsec:mrnp}. 
We consider two ways of inducing an uncertainty set from a prescribed region in the SRG plane, %
and discuss in details when and what necessary and sufficient SRG-based conditions for MRN can be established with respect to these region-induced uncertainties.

\subsection{Preliminaries on Davis-Wielandt Shells\label{subsec:prelim-dw-shell}}
To solve the MRN problem, we will be relying on another convex graphical tool in the three-dimensional space, called the Davis-Wielandt (DW) shell, that projects into the SRG and $\theta$-SRG. The DW shell of a matrix $M \in \cF^{n\times n}$, which resides in $\cF\times [0,\infty)$, is defined as:
\begin{align}
    \label{eq:dw-defn}
    \dwshell(M):= 
    \left\{\left(\frac{\iprod(u)(Mu)}{\|u\|^2},\frac{\|Mu\|^2}{\|u\|^2}\right): u\neq \mato\right\}.
\end{align}
The notion of DW shell dates back to the early work of Wielandt (normal case) and Davis (general case) \cite{wielandtEigenvaluesSumsNormal1955,davisShellHilbertspaceOperator1968,davisShellHilbertspaceOperator1970} and is viewed as a form of generalised numerical range \cite{gustafsonNumericalRangeField2012,liDavisWielandtShellsOperators2008} in linear algebra. Its implicit usage in control appears in \cite{jonssonScalableRobustStability2010,kaoCharacterizationRobustStability2009} and is explicitly introduced and developed under the integral quadratic constraints (IQCs) framework in \cite{lestasNetworkStabilityGraph2011,lestasLargeScaleHeterogeneous2012}. It has found a series of applications recently in \cite{zhaoWhenSmallGain2022,liangFeedbackStabilityMixed2025,zhangLocalStabilityCongestion2025} and appears to have intimate tie with the SRG and its variants \cite{zhangPhantomDavisWielandtShell2025a}.
Specifically, \cite{zhangPhantomDavisWielandtShell2025a} reveals that $\srg(M)$ can be obtained from $\dwshell(M)$ via a two-step projection: 
\begin{enumerate}
    \item project each point in $\dwshell(M)$ horizontally onto the paraboloid $\parab[1]$ along the imaginary axis in both directions (horizontal arrows in \cref{fig:dw-srg}); 
    \item project the shadows on the paraboloid from 1) vertically onto the SRG plane (vertical arrows in \cref{fig:dw-srg}). 
\end{enumerate}
Then the variant $\srg[\theta](M)$ is obtained by changing the direction (axis) in 1) to the axis orthogonal to the $\theta$-axis. 
We denote the composition of the abovesaid projections by $\proj[\theta]$, which acts pointwise on $\epi(\parab[1]) \subseteq \cF\times [0,\infty)$ as%
\footnote{Note that in \cite{zhangPhantomDavisWielandtShell2025a}, the notation $\proj[\theta]$ refers to only the half of the first projection. In this paper, for notational simplicity, it denotes the composition of the two projections, which maps from $\epi(\parab[1])$ to the SRG plane.}
\begin{align*}
    \proj[\theta](z,\nu) :=  e^{\ii*\theta} \left(\tRe(e^{-\ii*\theta}z)\pm \ii*\sqrt{\nu - \tRe(e^{-\ii*\theta}z)^2}\right).
\end{align*}
Note that the DW shell always resides in $\epi(\parab[1])$, which we refer to as the \emph{DW space}.
Each point in the DW space is mapped by $\proj[\theta]$ to a  $\theta$-conjugate pair in the SRG plane,
and $\proj[\theta] \dwshell(M) = \srg[\theta](M)$.
Readers may refer to \cite{liDavisWielandtShellsOperators2008,liEigenvaluesSumMatrices2008,lestasNetworkStabilityGraph2011,zhangPhantomDavisWielandtShell2025a} for an expository study on the DW shells and their connection to, e.g., quadratic constraints, SRGs, and numerical ranges. In the following, we only recite those relevant to the specific problems herein.  
\begin{figure}
    \centering
    \subfloat[{$\proj[0] \dwshell(M) = \srg(M)$}]{
        \adjincludegraphics[Clip = {0.14\width} {.16\height} {.2\width} {.05\height}, height=4cm]{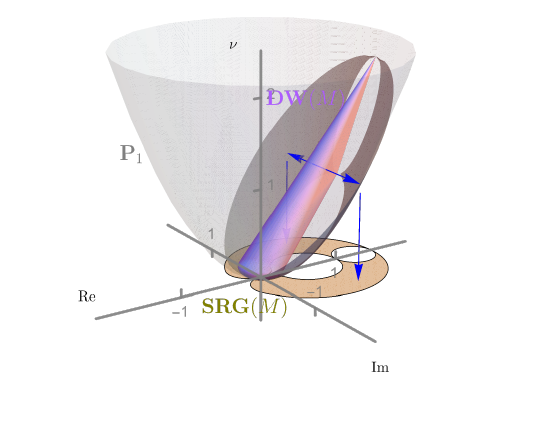}}
    \subfloat[{$\proj[\theta] \dwshell(M) = \srg[\theta](M)$}]{
        \adjincludegraphics[Clip = {0.14\width} {.16\height} {.2\width} {.05\height}, height=4cm]{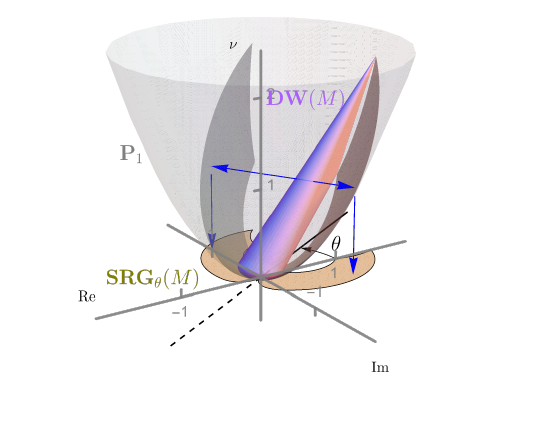}}
    \caption{The $\theta$-SRG can be obtained as a two-step projection of the DW shell.\label{fig:dw-srg}}\vspace{-.6cm}
\end{figure}

The inverse DW shell of $M$ is defined as 
\begin{align*}
    \invdwshell(M):=\invmap(\dwshell(M)\setminus\{(0,0)\})
\end{align*} where $\invmap(z,\nu):=(\cj(z)/\nu, 1/\nu)$, and its inverse $\theta$-SRG is defined as the image of $\invdwshell(M)$ under the same projection that is used to obtain $\srg[\theta](M)$ from $\dwshell(M)$.
In cases where $M$ is nonsingular, we have 
\begin{align*}
    \invdwshell(M) &= \dwshell(M\inv),\\
    \invsrg[\theta](M) &= \srg[\theta](M\inv) = (\srg[-\theta](M))\inv.
\end{align*}
Note that the origin is contained in its DW shell (or $\theta$-SRG) if and only if $M$ has a nontrivial kernel. In such cases, its inverse DW shell (or $\theta$-SRG) becomes unbounded. The origin is removed from its DW shell before being fed into $\invmap$ when defining its inverse shell. This is purely for technical convenience and it is the asymptotic behaviour near the infinity that matters the most in our context. 
Both the DW shell and its inverse shell are contained in $\epi(\parab[1])$, a uniform paraboloidal bound arising from the Cauchy-Schwarz inequality.
With the DW shell and its inverse shell defined, the following separation condition \cite{lestasLargeScaleHeterogeneous2012,liEigenvaluesSumMatrices2008,zhangPhantomDavisWielandtShell2025a} is central for characterising MRN:
\begin{lemma}[DW Separation Condition] \label{lem:dw-sep-ext}
    For given matrices $G, \Delta\in \cF^{n\times n}$, the matrix $\mati + GU\ct \Delta U$ is nonsingular for all unitary $U$ if and \emph{only if} $\invdwshell(-G)\cap\dwshell(\Delta)=\emptyset$.
\end{lemma}

For a set $\symset$, we denote by $\dwshell(\symset)$ (or $\invdwshell(\symset)$) the union of DW shells (or inverse DW shells) of all matrices in $\symset$, i.e., 
\begin{align*}
    \dwshell(\symset) &:= \cup_{\Delta\in\symset} \dwshell(\Delta),\\
    \invdwshell(\symset) &:= \cup_{\Delta\in\symset} \invdwshell(\Delta).
\end{align*} 
The set $\symset$ is said to be \emph{unitarily invariant} if $U\ct X U\in \symset$ for all $X\in \symset$ and unitary $U$.
For unitarily invariant sets, \cref{lem:dw-sep-ext} can be adapted to the following robust version:
\begin{corol}\label{corol:mrn-dw}
    Given a unitarily invariant matrix set $\symset$ and a matrix $G$ of compatible size, the following are equivalent:
    \begin{enumerate}
        \item MRN holds for $G$ with respect to $\symset$;
        \item $\invdwshell(-G)\cap \dwshell(\symset) = \emptyset$;
        \item $\dwshell(G) \cap \invdwshell(-\symset) = \emptyset$.
    \end{enumerate}
\end{corol}
\begin{proof}
    For each $\Delta \in \symset$, 2) ensures that $\invdwshell(-G)\cap\Delta = \emptyset$. By the ``if'' part of \cref{lem:dw-sep-ext}, this implies that $\mati + G\Delta$ is nonsingular. Recalling \cref{defn:mrn}, we conclude that 2) implies 1). Conversely, since $\symset$ is unitarily invariant, $\Delta\in \symset$ implies that $U\ct \Delta U\in\symset$ for all unitary $U$. By the ``only if'' part of \cref{lem:dw-sep-ext}, it follows that 1) implies 2). Hence, 1) and 2) are equivalent. The equivalences between 1) and 3) can be shown similarly by swapping the roles of $G$ and $\Delta$.
\end{proof}
The above DW separation condition provides a general 3-D characterisation of MRN with respect to a unitarily invariant set.
We are now in a position to solve the MRN problems on the SRG plane.

\subsection{Robust Conditions: Disc, Conic, and Sectorial Regions\label{subsec:mrs-special}}
In this section, we present necessary and sufficient SRG-based conditions for MRN with respect to uncertainties induced from simple geometric regions (disc, cone, sector) determined by gain bound and/or angular intervals. The main purpose is to answer our initial questions and to demonstrate the key ideas in a relatively simple setting. The proofs are omitted, as the entire section can be viewed as a special case of the results developed in the next section.

Note that $\dwshell(\Delta)$ is invariant under unitary similarity transformation on $\Delta$, i.e., $\dwshell(U\ct* \Delta U) = \dwshell(\Delta)$ for all unitary $U$, and hence $\srg[\theta](\Delta)$ is unitarily invariant as well. 
Then it follows that any matrix uncertainty set defined in terms of the $\theta$-SRG (or the DW shell) being contained within a prescribed 2-D (or 3-D) region, including  $\gasys_\gamma, \cosys[\alpha,\beta]$, and $\sesys_\gamma[\alpha,\beta]$, is unitarily invariant. Therefore, we can simply apply \cref{corol:mrn-dw} to the disc, conic, and sectorial uncertainty sets. This leads us to the following characterisations:
\begin{prop} \label{prop:ga-co-sec-char}
     Given $\gamma>0$, $\alpha,\beta\in \rF$ with $\alpha \leqsl \beta$, the following equalities hold when the matrix dimension is at least $2$:
     \begin{align}
        \dwshell(\gasys_\gamma) &= \epi(\parab[1]) \cap \{(z,\nu): \nu \leqslant \gamma^2\}; \label{eq:ga-char}\\
        \dwshell(\cosys[\alpha,\beta]) &= \parahull((\seg[\alpha,\beta]\times \{1\})); \label{eq:ang-char}\\
        \dwshell(\sesys_\gamma[\alpha,\beta]) &= 
        \dwshell(\gasys_\gamma) \cap \dwshell(\cosys[\alpha,\beta]). \label{eq:ga-ang-char}
     \end{align}
     When the matrix dimension is $1$ (i.e., in the scalar case), the left-hand sides of \cref{eq:ga-char,eq:ang-char,eq:ga-ang-char} degenerate to the intersections of their corresponding right-hand sides with the surface $\parab[1]$.
\end{prop}
\begin{figure}[h]
    \centering
    \subfloat[{Disc}]{
        \adjincludegraphics[Clip = {0.14\width} {.16\height} {.05\width} {.06\height}, width=.33\columnwidth]{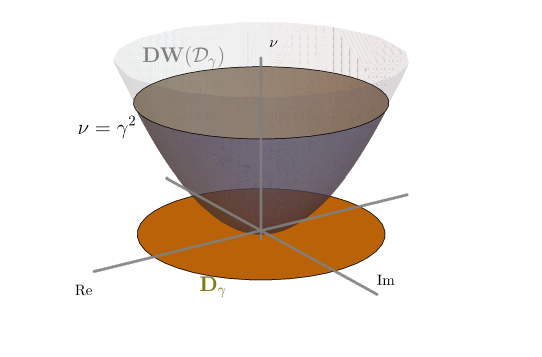}\hspace{-2em}}
    \subfloat[{Cone}]{
        \adjincludegraphics[Clip = {0.14\width} {.16\height} {.05\width} {.06\height}, width=.33\columnwidth]{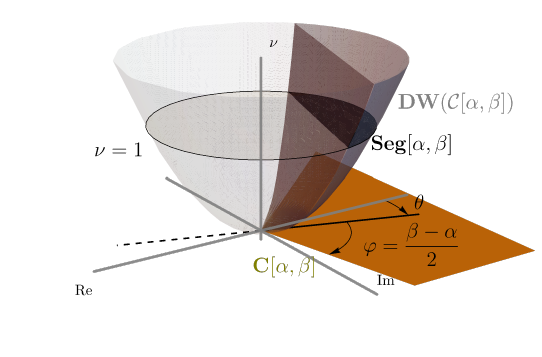}}
    \subfloat[{Sector}]{
        \adjincludegraphics[Clip = {0.14\width} {.16\height} {.05\width} {.06\height}, width=.33\columnwidth]{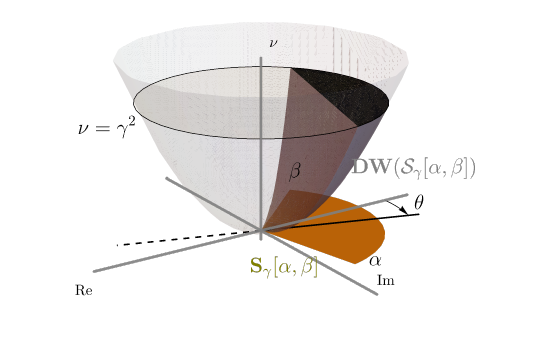}}
    \caption{Precise characterisations of the DW shell unions of disc, conic, and sectorial uncertainties.\label{fig:ga-ang-plots}}\vspace{-1em}
\end{figure}

Combining
\cref{corol:mrn-dw} %
with
the precise characterisation in \cref{prop:ga-co-sec-char}, and bridging the 2-D and 3-D conditions via the projection $\proj[\theta]$, the following three results immediately follow.
\begin{corol}[{Disc Uncertainty}] \label{corol:ma-small-gain}
    Given $\gamma>0$ and $G\in \cF^{n\times n}$, the following statements are equivalent:
    \begin{enumerate}
        \item MRN holds for $G$ with respect to $\gasys_\gamma$;
        \item $\invdwshell(-G) \cap (\epi(\parab[1]) \cap \{(z,\nu): \nu \leqslant \gamma^2\}) = \emptyset$;
        \item $\invsrg[](-G) \cap \disc[\gamma] = \emptyset$;
        \item $\srg(G) \subseteq \sint(\disc[1/\gamma])$, or simply $\|G\|_2 < 1/\gamma$.
    \end{enumerate}
\end{corol}
The above recovers the well-known small-gain condition (e.g.,~\cite[Thm.~9.1]{zhouRobustOptimalControl1995}) at the matrix level, which is both necessary and sufficient.

Before proceeding to the next result, we first revisit an existing angular concept, which we term as $\theta$-angle. 
The $\theta$-angle of $G$, denoted by $\angmax[\theta](G)$, is the largest absolute argument deviation from $\theta$: 
\begin{align*}
\angmax[\theta](G):= \sup\,\{|\angle z - \theta| : z\in \srg[\theta](G)\},
\end{align*}
where $\angle z$ is always confined to $[\theta-\pi,\theta+\pi)$ and hence $\angmax[\theta](G)$ falls within $[0,\pi]$.
This angular quantity (where $\theta$ is often fixed to $0$) appears in various contexts with different names, e.g., \cite{gustafsonAntieigenvalues1994,paulComputationAntieigenvaluesBounded2015,chenSingularAngleNonlinear2025,chenCyclicSmallPhase2025,zhangPhantomDavisWielandtShell2025a}.

\begin{corol}[{Conic Uncertainty}] \label{corol:ma-small-angle}
    Given $\alpha,\beta \in \rF$ with $\alpha\leqsl\beta$ and $G\in \cF^{n\times n}$, the following statements are equivalent:
    \begin{enumerate}
        \item MRN holds for $G$ with respect to $\cosys[\alpha,\beta]$;
        \item $\invdwshell(-G)\cap \parahull((\seg[\alpha,\beta]\times \{1\})) = \emptyset$;
        \item $\invsrg[(\alpha+\beta)/2](-G) \cap \cone*[\alpha,\beta] = \emptyset$;
        \item $\srg[-(\alpha+\beta)/2](G)\subseteq \{0\}\cup\,\sint(\cone*[-\alpha-\pi, \pi - \beta])$, or simply $\angmax[-(\alpha+\beta)/2](G) < \pi - (\beta-\alpha)/2$.
    \end{enumerate}
\end{corol}
This unveils the necessity of angle-type conditions, whose sufficient side has been explored extensively in \cite{chenSingularAngleNonlinear2025,chenCyclicSmallPhase2025}.  

\begin{corol}[{Sectorial Uncertainty}] \label{corol:ma-sect}
    Given $\gamma>0$, $\alpha,\beta \in \rF$ with $\alpha\leqsl\beta$, and $G\in \cF^{n\times n}$, the following statements are equivalent:
    \begin{enumerate}
        \item MRN holds for $G$ with respect to $\sesys_\gamma[\alpha,\beta]$;
        \item $\invdwshell(-G)\cap (\dwshell(\gasys_\gamma) \cap \dwshell(\cosys[\alpha,\beta])) = \emptyset$;
        \item $\invsrg[(\alpha+\beta)/2](-G) \cap \sector[\gamma][\alpha,\beta] = \emptyset$;
        \item $\srg[-(\alpha+\beta)/2](G) \subseteq \sint((\disc[1/\gamma]\cup \cone*[-\alpha-\pi,\pi-\beta]))$.
    \end{enumerate}
\end{corol}
This constitutes a mixed gain-angle type robust condition. It is worth noting that, in the last condition, the allowable region for the SRG variant is the union of the regions corresponding to the gain and angle conditions individually.

\subsection{Robust Conditions: General Regions\label{subsec:mrs-generalise}}
From above, we observe two elements that are essential for building a pathway from \cref{corol:mrn-dw} to a SRG-based MRN condition without introducing conservatism (i.e., with necessity): (1) a precise characterisation of the union of DW shells of uncertainties prescribed by the region in the SRG plane; (2) the confinement of the non-paraboloidal boundary of the DW shell union to a set of cylindrical surfaces who share a common axis in the SRG plane.
Let $\symreg$ be a general closed region in the SRG plane. 
    We introduce the notions of $\theta$-symmetric part and $\theta$-symmetric cover of $\symreg$ (as illustrated in \cref{fig:sym-inf-sup}), and accordingly formalise the $\theta$-symmetry property, which will be of frequent use subsequently:
    \begin{defn}[{{$\theta$}-symmetry}]
        Given $\theta\in \rF$ and a region $\symreg \subseteq \cF$, 
        \emph{the $\theta$-symmetric part} of $\symreg$, denoted by $\infreg[\theta]$, is the \emph{the largest subregion }of $\symreg$ that is symmetric about the $\theta$-axis:
        \begin{align}
            \infreg[\theta] 
            :=  \{z  : z \in \symreg \text{ and }  z_\theta \in \symreg\}. \label{def:sym-part}
        \end{align}
        The \emph{$\theta$-symmetric cover} of $\symreg$, denoted by $\supreg[\theta]$, is the \emph{the smallest supregion} of $\symreg$ that is symmetric about the $\theta$-axis:
        \begin{align}
        \supreg[\theta]
            := \{z : z \in \symreg\text{ or } z_\theta \in \symreg\}. \label{def:sym-cover}
        \end{align} 
        The region $\symreg$ is said to be \emph{$\theta$-symmetric} if $\infreg[\theta] = \symreg = \supreg[\theta]$.
    \end{defn}
    \begin{figure}
        \centering
        \begin{minipage}[c]{0.53\columnwidth}
            \centering
            \adjincludegraphics[Clip = {0.15\width} {0} {0} {.18\height}, width=\linewidth]{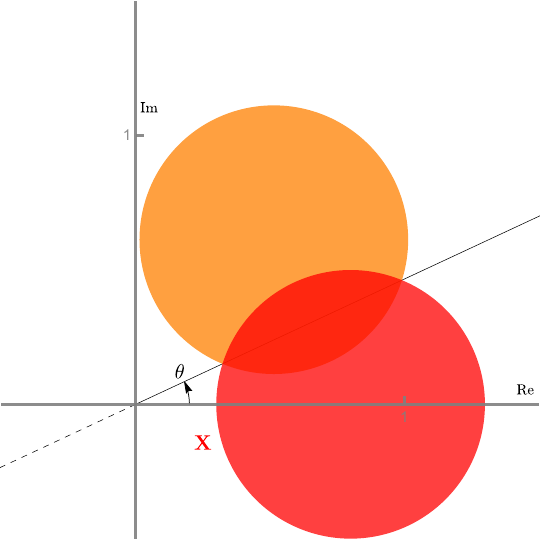}
        \end{minipage}
        \hfill
        \begin{minipage}[c]{0.45\columnwidth}
            \centering
            \fbox{\adjincludegraphics[Clip = {0.15\width} {0} {0} {.45\height},width=.7\linewidth]{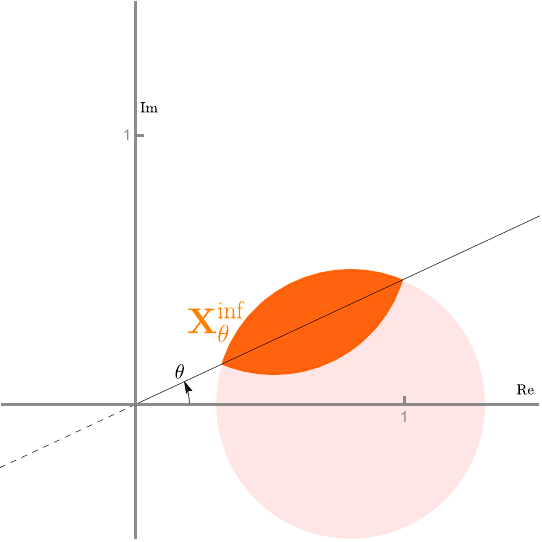}}
            
            \vspace{.2em} 
            
            \fbox{\adjincludegraphics[Clip = {.15\width} {0} {0} {.18\height}, width=.7\linewidth]{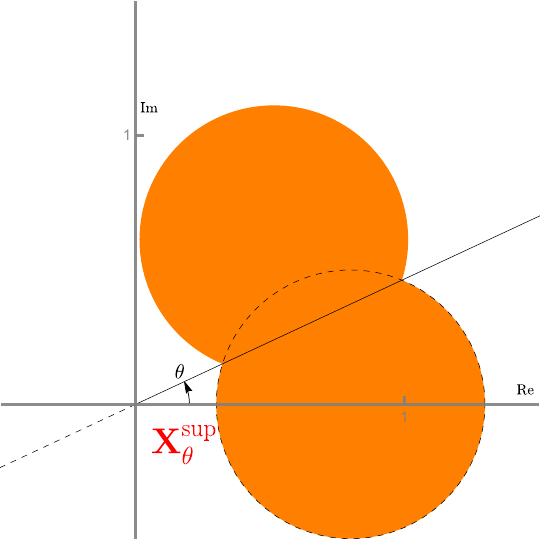}}
        \end{minipage}
        \caption{The $\theta$-symmetric part and cover of a region $\symreg$ (red). \label{fig:sym-inf-sup}} \vspace{-3em}
    \end{figure}
We consider a general uncertainty set induced by a general $\symreg$ as:
\begin{align}
    \symset :=& \{\Delta \in \cF^{n\times n}: \exists\, \vartheta \in \rF \nonumber \\
    &\hspace{6em} \text{ such that } \srg[\vartheta](\Delta) \subseteq \symreg\}. 
    \label{def:general-uncert}
\end{align}
The above definition assumes that there is no prior knowledge about $\vartheta$ but only a region in the SRG plane that results from SRG variants of the uncertainty.

The dependence on the existence of a non-prescribed $\vartheta$ in defining the uncertainty set $\symset$ may appear awkward in practice. 
We can also fix a $\theta\in\rF$, then the $\theta$-uncertainty set induced by the region $\symreg$ is defined as
\begin{align}
    \label{def:theta-uncert}
    \symset_\theta := \{\Delta \in \cF^{n\times n} : \srg[\theta](\Delta) \subseteq \symreg\}.
\end{align}
Recall that the $\theta$-SRG is inherently $\theta$-symmetric. 
    Thus, only the symmetric part $\infreg[\theta]$ acts as an active region bound in (\ref{def:theta-uncert}), that is 
        $\symset_\theta = \infreg*[\theta] := \{\Delta : \srg[\theta](\Delta) \subseteq \infreg[\theta]\}$. 
With $\theta$-uncertainty sets, we can also reformulate (\ref{def:general-uncert}) as  $\symset = \cup_{\theta \in \rF}\symset_\theta$.
We shall next consider the MRN problem with respect to both $\symset$ and $\symset_\theta$, and also discuss when these two problems (sets) coincide.
Denote the smallest and largest magnitudes among all points in $\symreg$ by:
        \begin{align*}
            \gmin(\symreg) := \inf \{|z| : z \in \symreg\},\quad
            \gmax(\symreg) := \sup \{|z| : z \in \symreg\}.
        \end{align*}
Now we provide a precise characterisation of the union of DW shells with $\Delta$ ranging over $\symset$ and $\symset_\theta$, respectively.
\begin{prop} \label{prop:dw-union-general}
    Given a closed region $\symreg$ in the SRG plane and $\theta\in \rF$.
    The following equalities hold when the matrix dimension is at least 2:
    \begin{align*}
        \begin{array}{@{}l@{\,}l@{}c@{}l}
            \dwshell(\symset_\theta) &=& 
                \displaystyle\bigcup_{\gmin(\infreg[\theta])\leqsl \gamma\leqsl \gmax(\infreg[\theta])}&
                \left(\chordhull((\disc[\gamma] \cap \infreg[\theta]))\times \{\gamma^2\} \vphantom{\chordhull(\{\disc[\gamma] \cap \symreg_\theta\})}\right), \\
            \dwshell(\symset) &=&
                \displaystyle\bigcup_{\gmin(\symreg)\leqsl \gamma\leqsl \gmax(\symreg)}
                &\left((\convhull((\disc[\gamma] \cap \symreg))) \times \{\gamma^2\}\right).
        \end{array}    
    \end{align*}
    When the matrix dimension is $1$ (i.e., in the scalar case), the left-hand sides above degenerate to the intersections of their corresponding right-hand sides with the surface $\parab[1]$.
\end{prop}
\begin{proof}
    See Appendix~\ref{pf:dw-union-general}.
\end{proof}

With the union of DW shells of uncertainties of interest precisely characterised, our next results identify situations where \cref{corol:mrn-dw} reduces to a condition on the SRG plane with necessity.
\begin{theorem}[{Matrix $\theta$-Uncertainty}] \label{thm:mat-theta}
    Given a $\theta\in\rF$ and a closed $\theta$-symmetric region $\symreg$ in the SRG plane, 
        the MRN holds for $G$ with respect to the corresponding uncertainty set $\symset_\theta$ 
            if and only if 
        $\invsrg[\theta](-G)\,\cap\,\symreg = \emptyset$.
\end{theorem}
\begin{proof}   
    See Appendix~\ref{pf:theta-uncert}.
\end{proof}
\begin{remark}
As mentioned, it generally holds that $\symset_\theta = \infreg*[\theta]$, where $\infreg*[\theta]$ is induced by the $\theta$-symmetric region $\infreg[\theta]$. 
    Therefore, regardless of whether the given region $\symreg$ is $\theta$-symmetric, the necessary and sufficient condition for MRN with respect to $\symset_\theta$ is the separation between $\invsrg[\theta](-G)$ and $\infreg[\theta]$.  
On the other hand, in practice it may be safer to consider the overestimate $\supreg*[\theta]$ of the uncertainty. For example, the prescribed region might be directly from data and any asymmetry may arise from numerical errors or data infidelity. Due to the $\theta$-symmetry of $\invsrg[\theta](-G)$ (refer to \cite[Tab.~III]{zhangPhantomDavisWielandtShell2025a}), the condition $\invsrg[\theta](-G)\cap \symreg$ is equivalent to $\invsrg[\theta](-G) \cap \supreg[\theta] = \emptyset$, the latter being the necessary and sufficient condition for MRN with respect to $\supreg*[\theta]$. 
These relationships are summarised in \cref{fig:theta-uncert-conditions}. All the conditions coincide if and only if $\symreg$ is $\theta$-symmetric.
    \begin{figure}
    \centering
    \begin{tikzpicture}[
        gbox/.style={
            draw,
            rectangle,
            rounded corners,
            align=center,
            minimum width=4cm,
            minimum height=.8cm
        },
        rbox/.style={
            draw,
            rectangle,
            rounded corners,
            align=center,
            minimum width=3.2cm,
            minimum height=1.0cm
        },
        arrow/.style={->, double, double distance = 2pt}
    ]
        \node (ghead) at (0,0) {\textbf{SRG Conditions}};
        \node (rhead) at (4.4,0) {\textbf{MRN Properties}};

        \node (g1) [gbox, below=0.8cm of ghead.center, anchor = center]  {$\invsrg[\theta](-G)\cap\supreg[\theta] = \emptyset$};
        \node (g2) [gbox, below=2.3cm of ghead.center, anchor = center] {$\invsrg[\theta](-G)\cap\symreg = \emptyset$};
        \node (g3) [gbox, below=3.8cm of ghead.center, anchor = center] {$\invsrg[\theta](-G)\cap\infreg[\theta] = \emptyset$};

        \node (r1) [rbox, below=0.8cm of rhead.center, anchor = center] {MRN holds for $G$ \\ w.r.t. $\supreg*[\theta]$};
        \node (r2) [rbox, below=2.3cm of rhead.center, anchor = center] {MRN holds for $G$ \\ w.r.t. $\symset_\theta$};
        \node (r3) [rbox, below=3.8cm of rhead.center, anchor = center] {MRN holds for $G$ \\ w.r.t. $\infreg*[\theta]$};

        \node at ($(g1)!0.55!(r1)$) {$\Longleftrightarrow$};
        \node at ($(g3)!0.55!(r3)$) {$\Longleftrightarrow$};

        \node at ($(g1)!0.5!(g2)$) {$\Updownarrow$};
        \node at ($(g2)!0.5!(g3)$) {$\Downarrow$};
        \node at ($(r1)!0.5!(r2)$) {$\Downarrow$};
        \node at ($(r2)!0.5!(r3)$) {$\Updownarrow$};
    \end{tikzpicture}
    \caption{Relationships among various SRG conditions and MRN properties with respect to $\theta$-uncertainty sets.\label{fig:theta-uncert-conditions}}
    \vspace{-2em}
    \end{figure}
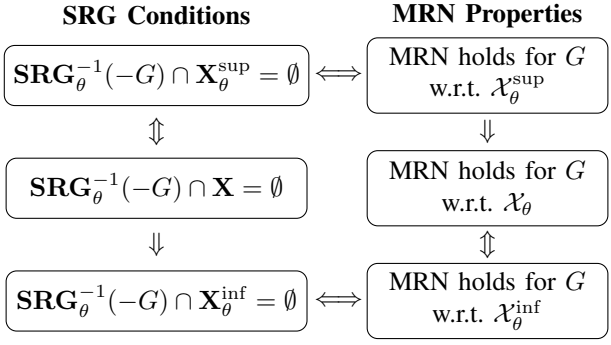
\end{remark}

By replacing $\symreg$ with $\disc[\gamma],\cone*[\alpha,\beta]$, and $\sector[\gamma][\alpha,\beta]$, and setting $\theta = 0$ for the disc case and $\theta = (\alpha+\beta)/2$ for the conic and sectorial cases,  
the equivalences between items 1) and 3) in our previous results, Corollaries \ref{corol:ma-small-gain}-\ref{corol:ma-sect}, follow directly from \cref{thm:mat-theta}. The equivalences with item 4) follow from a straightforward pointwise argument by using \cite[Prop.~2 and Tab.~III]{zhangPhantomDavisWielandtShell2025a} together with the closedness of the disc, conic, and sectorial regions.

Next, we move on to the case of general uncertainties as defined in \cref{def:general-uncert}. Before showing the main results, we introduce the notion of $\theta$-circular connectedness and the associated circular hull operation that will streamline our statements. 
    \begin{defn}[{$\theta$}-circular connectedness] \label{defn:theta-cnt}
        Given $\theta\in \rF$, a region $\region\subseteq\cF$ is said to be \emph{$\theta$-circularly connected} if, for each $\gamma\in [\gmin(\region),\gmax(\region)]$, the intersection of its circular cross section $\disc[\gamma]\cap \region$ with each of the two closed half-planes separated by the $\theta$-axis is either empty or \emph{connected}.
    \end{defn}

    As an example, the broken-heart-shaped region $\region$ in \cref{fig:theta-cir} is $0$-circularly connected but not $\theta$-circularly connected for the value of $\theta$ indicated in the figure.

For a closed region $\region \subseteq \cF$, a $\theta\in\rF$, and some $\gamma\in [\gmin(\region),\gmax(\region)]$, the circle-wise $\theta$-angle of $\region$ are defined as
\begin{align*}
    \angmin[\theta](\region; \gamma) &:= \inf\,\{|\angle z - \theta| : z\in \region \cap \disc[\gamma]\};\\
    \angmax[\theta](\region; \gamma) &:= \sup\,\{|\angle z - \theta| : z\in \region \cap \disc[\gamma]\},
\end{align*} where $\angle z$ takes value in $[\theta-\pi,\theta+\pi)$.
Under this definition, the previously defined $\theta$-angle of an operator can be rewritten as $\angmax[\theta](G) = \sup_{\gamma \in [\svmin(G),\svmax(G)]}\angmax[\theta](\srg[\theta](G); \gamma)$.
The $\theta$-circular hull of $\region$, denoted by $\circhull(\region)$, is then defined as:
\begin{align*}
    \bigcup_{\gmin(\region)\leqsl\gamma\leqsl\gmax(\region)} \{\gamma e^{\ii* (\theta+\delta) }: |\delta| \in [\angmin[\theta](\region;\gamma),\angmax[\theta](\region;\gamma)]\}.
\end{align*}

The following observation is straightforward from definitions:
\begin{obs} \label{obs:circhull}
    Given a $\theta\in\rF$,
    a closed region $\region$ in the SRG plane is $\theta$-symmetric and $\theta$-circularly connected if and only if $\region = \circhull(\region)$.
    If either $\theta$-symmetry or $\theta$-circularly connectedness is absent, $\region$ is a \emph{proper} subset of $\circhull(\region)$.
\end{obs}

\begin{figure}
    \centering
    \adjincludegraphics[Clip={.15\width} {.14\height} {.05\width} {.04\height}, height=4cm]{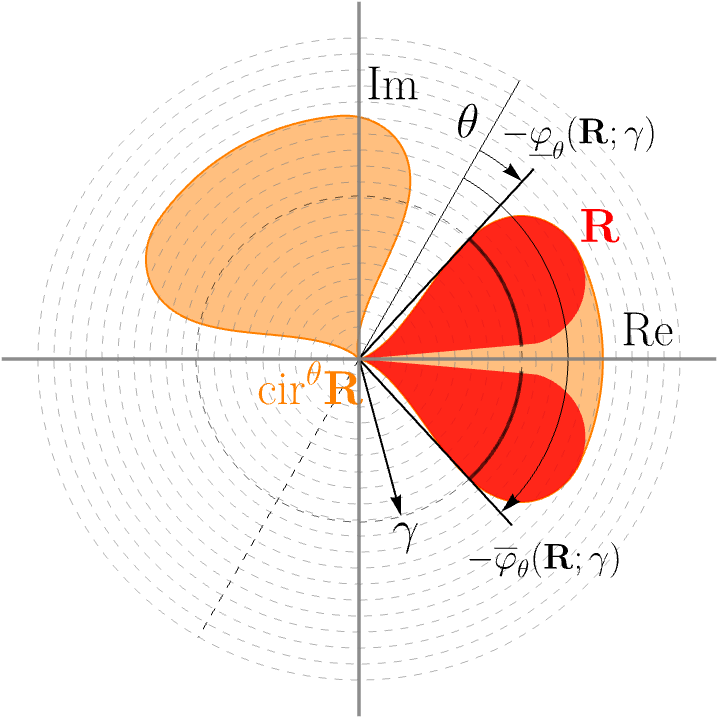}
    \caption{The $\theta$-circular hull (orange) of a broken heart $\region$ (red). \label{fig:theta-cir}}\vspace{-3em}
\end{figure}

In the case of general uncertainties, separation from the $\theta$-circular hull of the prescribed region provides %
a general sufficient condition. Meanwhile, $\theta$-symmetry and $\theta$-circular connectedness are key properties that ensure the separation is also universally necessary. 
\begin{theorem}[{Matrix General Uncertainty}] \label{thm:mat-general}
    Let the matrices under consideration be of size at least 2. Given a closed region $\symreg$ in the SRG plane and an arbitrary $\theta \in \rF$,
        the MRN holds for $G$ with respect to the uncertainty set $\symset$ 
            if
        $\invsrg[\theta](-G)\,\cap\,\circhull(\symreg) = \emptyset$.
    The graphical condition becomes necessary if and only if $\symreg$ is $\theta$-symmetric and $\theta$-circularly connected.
\end{theorem}
\begin{proof}
    See Appendix~\ref{pf:general-uncert}.
\end{proof}

Comparing \cref{thm:mat-theta} and \cref{thm:mat-general}, 
    when the prescribed uncertainty region is both $\theta$-symmetric and $\theta$-circularly connected, 
    the corresponding graphical conditions on the SRG plane coincide and are both necessary and sufficient for MRN with respect to their associated uncertainty sets. 
    It is therefore not surprising that $\symset = \symset_\theta$ in this case.
\begin{prop} \label{prop:sym-set-char}
    Given a $\theta\in\rF$ and a region $\symreg$ in the SRG plane, 
        the $\theta$-uncertainty set $\symset_\theta$ is a subset of the uncertainty set $\symset$.
        When the matrices under consideration are of size at least 2, these two sets coincide if and only if $\symreg$ is $\theta$-symmetric and $\theta$-circularly connected.
\end{prop}
\begin{proof}
    See Appendix~\ref{pf:sym-set-char}.
\end{proof}
\begin{remark}
    Note that the disc, conic, and sectorial regions are all $\theta$-symmetric and circularly connected with the $\theta$-axis being the axis of symmetry (so $\theta$ can be any real number for disc regions).
\end{remark}
\begin{remark}
    In the scalar case, as shown in the proof of \cref{prop:dw-union-general}, $ \symset_\theta=\infreg[\theta]$ and $\symset = \symreg$. Clearly, $\symset_\theta = \symset$ if and only if $\symreg$ is $\theta$-symmetric. In this case, $\theta$-circular connectedness becomes irrelevant.    
\end{remark}
\begin{remark} \label{rem:justify-dw}
    It is always possible to suppress the entire intermediate step via DW shell by slightly adapting the arguments that lead to \cite[Thm.~IV.1]{zhangPhantomDavisWielandtShell2025a}.
    However, we believe that the DW shell is a well-thought convex%
    \footnote{Convexity could only fail for matrices of size 2, where the DW shell could be literally an (ellipsoidal) shell without interior. However, this won't cause any problem when considering shell separations.} 
    black box in higher dimension that encapsulates unnecessary technicality and offers many clean yet powerful interfaces for us to visualise, interpret, and develop results in lower dimensions. This echos with many situations in control and optimisation where lifting to higher dimensions presents a conceptually easier viewpoint to dismantle seemingly complicated problems in lower dimensions.
\end{remark}

\section{SYSTEM ROBUST STABILITY\label{sec:system-ana}}
We now consider the robust stability problem of feedback interconnection as formulated in \cref{subs:rs-formulation}.

The stability of the negative feedback interconnection as in \cref{fig:sys-diagram} will be established through a standard argument based on generalised Nyquist criterion \cite{desoerGeneralizedNyquistStability1980}.
For fixed $G, \Delta \in \rhinf^{n\times n}$, their feedback interconnection is stable if and only if
the eigenloci of the return ratio $G\Delta$ do not cross or encircle the $-1$ point. 
This can be ensured by requiring that the spectrum of $G(\ii*\omega)\Delta(\ii*\omega)$ does not intersect the interval $[-\infty,-1]$ on the real axis for all $\omega$:
\begin{align*}
    \det (\lambda \mati + G(\ii*\omega)\Delta(\ii*\omega)) \neq 0 \Forall \omega \in [0,\infty], \lambda \in [1,\infty]. 
\end{align*} 
The above condition is equivalent to:
\begin{align*}
    &\det (\mati + G(\ii*\omega) (\tau \Delta(\ii*\omega))) \neq 0 \\
    &\hspace{1.7cm}\Forall \omega \in [0,\infty], \tau \in [0,1]. 
    \numberthis \label{det-cond}
\end{align*}
\begin{defn} \label{def:contr-inv-set}
    A prescribed uncertain set $\uncert \subseteq \rhinf^{n\times n}$ of stable transfer matrices is said to be invariant under contractive scaling, or simply \emph{contractively invariant}, if it satisfies $\tau \uncert \subseteq \uncert$ for all $\tau \in [0,1]$. 
\end{defn}

For a contractively invariant uncertainty set $\uncert$, it follows from (\ref{det-cond}) and \cref{def:rs} that a given $G\in \rhinf^{n\times n}$ is robustly stable under negative feedback with respect to $\uncert$ if:
\begin{align*}
    &\det(I+G(\ii*\omega) \Delta(\ii*\omega)) \neq 0 \\
    &\hspace{1.7cm} \Forall \omega \in [0,\infty], \Delta \in \uncert. 
    \numberthis \label{cond-rs}
\end{align*}  
The form of (\ref{cond-rs}) at each frequency $\omega$ now matches that of the MRN property defined in \cref{defn:mrn}.
In the context of this and next sections,
we are concerned about $\uncert$ that is specified by frequencywise graphical region in the SRG plane.

Let $\symreg:[0,\infty]\rightarrow 2^{\cF}$ and $\theta:[0,\infty]\rightarrow \rF$ be region- and real-valued functions of frequencies, respectively.
Assume that the value of $\symreg$ is always closed. 
    Then the corresponding $\theta$- and general uncertainty sets are defined in a frequencywise fashion:
    \begin{align*}
        &\symset_\theta := \left\{\Delta \in \rhinf^{n\times n}: \srg[\theta(\omega)](\Delta(\ii*\omega)) \subseteq \symreg(\omega) \right.\\ 
        &\hspace{4.4cm}\left.\Forall\omega \in [0,\infty]\right\}, 
            \numberthis \label{sys-uncert-theta} \\
        &\symset := \left\{\Delta \in \rhinf^{n\times n}: \exists\, \vartheta(\omega)\in \rF \text{ such that }  \right.\\ 
        &\hspace{1.5em}\left.\srg[\vartheta(\omega)](\Delta(\ii*\omega)) \subseteq \symreg(\omega)\Forall\omega \in [0,\infty]\right\}. 
            \numberthis \label{sys-uncert-gen}
    \end{align*} 
Then we propose the following robust stability conditions:
\begin{theorem}[System $\theta$-Uncertainty] \label{thm:sys-theta}
    Assuming that $\symreg(\omega)$ is $\theta(\omega)$-symmetric for all $\omega\in [0,\infty]$,
    the feedback interconnection defined through $G\in \rhinf^{n\times n}$ is robustly stable with respect to $\symset_\theta$ if, for each $\omega \in [0,\infty]$, 
        \begin{align*}
            \invsrg[\theta(\omega)](-G(\ii*\omega)) \cap \starhull(\symreg(\omega)) = \emptyset, 
        \end{align*}
        or equivalently, 
        \begin{align*}
            \srg[-\theta(\omega)](G(\ii*\omega)) \subseteq \scomp(- (\starhull(\symreg(\omega))\setminus \{0\})\inv).
        \end{align*}
\end{theorem}
\begin{theorem}[System General Uncertainty] \label{thm:sys-general}
    The feedback interconnection defined through $G\in \rhinf^{n\times n}$ is robustly stable with respect to $\symset$ if there exists a function $\vartheta:[0,\infty]\rightarrow \rF$ such that, for each $\omega \in [0,\infty]$,
        \begin{align*}
            \invsrg[\vartheta(\omega)](-G(\ii*\omega)) \cap \starhull((\circhull[\vartheta(\omega)](\symreg(\omega)))) = \emptyset,
        \end{align*}
        or equivalently,
        \begin{align*}
            \srg[-\vartheta(\omega)](G(\ii*\omega)) \subseteq \scomp(- (\starhull((\circhull[\vartheta(\omega)](\symreg(\omega))))\setminus \{0\})\inv).
        \end{align*}
\end{theorem}
\begin{remark} \label{rem:sys-sym-circ}
    In the spirit of \cref{thm:mat-general},
    when the prescribed SRG region $\symreg(\omega)$ is frequencywise symmetric, we choose $\vartheta(\omega)$ in \cref{thm:sys-general} to be the angular function $\theta$ determining the axis of symmetry frequencywise.
    If, in addition, $\symreg(\omega)$ is $\theta(\omega)$-circularly connected, then this choice does not introduce extra conservatism. 
    In this case, the conditions from \cref{thm:sys-theta} and \cref{thm:sys-general} coincide.   
\end{remark}
\begin{myproof}[Proofs of \cref{thm:sys-theta,thm:sys-general}]
Both theorems follow from a frequencywise application of the corresponding matrix results, \cref{thm:mat-theta,thm:mat-general}, to (\ref{cond-rs}). We illustrate this by considering the case of $\theta$-uncertainty. 
Since the condition is stated frequencywise, we suppress the dependence on $\omega$ for notational simplicity.

Let $\symsettld_\theta$ denote the $\theta$-uncertainty set induced by $\starhull(\symreg)$.
Note that for any $\Delta \in \symset_\theta$ and $\tau \in [0,1]$, we have by \cite[Prop.~2]{zhangPhantomDavisWielandtShell2025a} that
$\srg[\theta](\tau \Delta) = \tau \srg[\theta](\Delta) \subseteq \tau \symreg \subseteq \starhull(\symreg)$. Therefore, $\cup_{\tau\in [0,1]}\tau \symset_\theta\subseteq \symsettld_\theta$ while the reverse containment does not hold in general.
It follows from \cref{thm:mat-theta} that $\invsrg[\theta](-G)\cap \starhull(\symreg) = \emptyset$ implies that MRN holds for $G$ with respect to $\symsettld_\theta$, and hence also with respect to $\cup_{\tau\in [0,1]}\tau \symset_\theta$, which is contractively invariant. Therefore, by (\ref{cond-rs}), the condition implies RS with respect to $\symset_\theta$.

Furthermore, by \cite[Prop.~2, Tab.~III]{zhangPhantomDavisWielandtShell2025a}, we can rewrite 
\begin{align*}
    &\invsrg[\theta](-G) = \inv(\srg[-\theta](-G)\setminus\{0\})\\
    &=(-\srg[-\theta-\pi](G)\setminus\{0\})\inv = (-\srg[-\theta](G)\setminus \{0\})\inv.
\end{align*}
A straightforward pointwise argument then yields the equivalence%
\footnote{Note that, by definition, $\invsrg[\theta](-G)$ cannot attain $0$. Therefore, $0$ can always be excluded from $\starhull(\symreg)$ in the condition without loss of generality.}
with $\srg[-\theta](G)\subseteq \scomp(-(\starhull(\symreg)\setminus\{0\})\inv)$.
The proof for \cref{thm:sys-general} follows similarly from \cref{thm:mat-general}.
\end{myproof}

Returning to the question of robust stability against disc, conic, and sectorial uncertainties $\gasys_\gamma, \cosys[\alpha,\beta]$, and $\sesys_\gamma[\alpha,\beta]$, their corresponding inducing regions are all frequency \emph{independent}, $\theta$-symmetric and $\theta$-circularly connected where $\theta$ can chosen arbitrarily for the disc $\disc[\gamma]$, and as $\frac{\alpha+\beta}{2}$ for the cone $\cone*[\alpha,\beta]$ and sector $\sector[\gamma][\alpha,\beta]$. 
It is also worth noting that their star hulls are all equal to themselves. Then, by \cref{thm:sys-theta}, we propose the following sufficient conditions:
\begin{corol}[Disc, Conic, Sectorial Uncertainties] \label{cor:ga-ph-sec}
    Let $\gamma >0 $ and $\alpha,\beta\in\rF$ with $\alpha \leqsl\beta$ be given. The feedback interconnection defined through $G(s)\in \rhinf^{n\times n}$ is robustly stable with respect to:
    \begin{enumerate}
        \item Disc uncertainty $\gasys_\gamma$ if 
            $\srg(G(\ii*\omega)) \subseteq \sint(\disc[1/\gamma])$ for each $\omega\in[0,\infty]$.
        \item Conic uncertainty $\cosys[\alpha,\beta]$ if, for each $\omega \in [0,\infty]$, 
            $\srg[-(\alpha+\beta)/2](G(\ii*\omega)) \subseteq \{0\}\cup\,\sint(\cone*[-\alpha-\pi,\pi-\beta])$.
        \item Sectorial uncertainty $\sesys_\gamma[\alpha,\beta]$ if, for each $\omega\in[0,\infty]$,
            $\srg[-(\alpha+\beta)/2](G(\ii*\omega)) \subseteq \sint((\disc[1/\gamma]\cup\cone*[-\alpha-\pi,\pi-\beta]))$.
    \end{enumerate}
\end{corol} 
\begin{remark}\label{rem:ga-ph-sec}
    The first condition is a small-gain condition. It can be interpreted as a bound on the $\hinf$-gain: $\|G(s)\|_\hinf < 1/\gamma$ where $\nrmhinf(G(s)):=\mathrm{ess\,sup}_{\omega\in \rF} \|G(\ii*\omega)\|_2$.
    Likewise, we define the $\hinf$-$\theta$-angle of a system by: 
    \begin{align*}
        \anghinf(G(s)):= \mathrm{ess\,sup}_{\omega\in[0,\infty]} \angmax[\theta](G(\ii*\omega)) \in [0,\pi].
    \end{align*} %
    The second condition can then be restated as a bound on the $\hinf$-angle: 
    $\anghinf[-(\alpha+\beta)/2](G(s)) < \pi - (\beta-\alpha)/2$.
    The third condition is a mixed gain-angle condition.
    It is well known that the first small-gain condition is also necessary (e.g., \cite[Thm.~9.1]{zhouRobustOptimalControl1995}). 
    In this paper, we establish the necessity of the remaining two conditions at the matrix level and conjecture that this necessity extends to systems as well.
    Furthermore, if this conjecture holds, we expect the ``only if'' part of \cref{thm:sys-theta} to hold whenever $\starhull(\symreg(\omega)) = \symreg(\omega)$ for each $\omega$.
    To ensure compatibility with the real rational uncertainty setting, we anticipate a need for mild assumptions on $\symreg(\omega)$, such as requiring $\symreg(0)$ and $\symreg(\infty)$ to contain either the nonnegative or nonpositive real axis.
\end{remark}
\begin{remark}
    For systems in $\rhinf^{n\times n}$, the positive real property can be characterised in terms of the $\hinf$-$\theta$-angle by $\anghinf[0](G(s))\leqsl \pi/2$, whereas the negative imaginary property is characterised by $\anghinf[-\pi/2](G(s))\leqsl\pi/2$.
\end{remark}

\section{REMARKS ON STATE-SPACE ROBUSTNESS CHARACTERISATIONS\label{sec:ss}}
The robust stability conditions in \cref{cor:ga-ph-sec} can be recast to quantify the robustness of $G$ against disc, conic, and sectorial uncertainties: $1/\nrmhinf(G)$ serves as a robustness measure against $\hinf$-gain-bounded uncertainties, and $\pi - \anghinf[-\theta](G)$ provides a robustness measure against $\hinf$-$\theta$-angle-bounded uncertainties. 
In general, quantifying robustness with respect to uncertainties induced by a class of geometric regions amounts to probing for the ``smallest'' such region that contains the frequencywise $\theta$-SRG of the given $G$. 

Consider a real rational stable LTI system $G\in\rhinf^{n\times n}$ with a minimal realisation $\ssreal$,
it is of general interest to investigate how the abovesaid process can be implemented in a numerically efficient way, which often relates to %
matrix inequality (MI) %
characterisations %
of the following form, denoted $\mi(X,Y,\tld(Z); \Theta(Z))$:
\begin{align*}
    &\milhs(X,Y,\tld({Z});\Theta(Z)) \preccurlyeq \mato, \qquad  Y \succcurlyeq \mato, 
\end{align*}
where
\begin{align*}
    \milhs(X,Y,\tld(Z);\Theta(Z))  = \left[\begin{smallmatrix}
        A  &  B \\ \mati  &  \mato
    \end{smallmatrix}\right]\ct
    \left[\begin{smallmatrix}
        \mato  & X+\ii* Y  \\ X-\ii* Y  &  \mato
    \end{smallmatrix} \right]
  \left[\begin{smallmatrix}
        A  & B \\ \mati  &  \mato
    \end{smallmatrix} \right]  +  \Theta(Z).
\end{align*}
Here, $\Theta(Z)$ denotes a Hermitian matrix associated with the geometric region %
parameterised by real parameters $Z$ 
and the matrix variables $X$, $Y$ are
Hermitian. %
$\tld(Z)$ denotes a subset of the parameters $Z$ that are treated as variables. When $Z$ is fixed, then $\tld(Z)$ is empty and will be surpressed. In situations where complex regions are probed via a family of regions defined by $\Theta(Z)$, this typically involves imposing constraints on $Z$ while varying a subset $\tld(Z)$. We therefore list $\tld(Z)$ together with $X$ and $Y$ to emphasise that these quantities are variables in the given context.
When the MI is linear in all variables, we insert `$\mathrm{L}$' before `$\mathrm{MI}$' to emphasise its linear nature.

In this section, we employ well-developed tools as in  \cite{guptaRobustStabilityAnalysis1996,bridgemanComparativeStudyInput2018,liangFeedbackStabilityMixed2025,yangSectoredRealLemma2025} and use discs (not necessarily centered at the origin) to bound the hyperbolic convex hull of $\theta$-SRGs of the nonnegative frequency response of $G$. Then by searching for the bounding discs in an organised way, we can develop a graphical robustness characterisation of $G$, which we term its $\theta$-$\varphi$-$\gamma$ profile.  For disc-boundedness, we %
can state the following result.

\begin{lemma} \label{lem:disc-bounded}
    Given $G\in \rhinf^{n\times n}$, $\theta\in\rF$, and $c \geqsl 0, r >0$, the following statements are equivalent:
    \begin{enumerate}
        \item $\srg[\theta](G(\ii*\omega)) \subseteq \disc[r](ce^{\ii*\theta})$ for all $\omega\in[0,\infty]$;
        \item $G$ satisfies the following quadratic constraint:
        \begin{align*}
            \begin{bNiceArray}{c}
                \mati \\ G(\ii*\omega)
            \end{bNiceArray}\ct
            \begin{bNiceArray}{cc}
                -(c^2-r^2)\mati_n & ce^{-\ii*\theta}\mati_n\\
                ce^{\ii*\theta}\mati_n & -\mati_n
            \end{bNiceArray}
            \begin{bNiceArray}{c}
                \mati \\ G(\ii*\omega)
            \end{bNiceArray} \succcurlyeq \mato,
        \end{align*} for all $\omega \in [0,\infty]$.
        \item $\lmi(X, Y; \mdisc(\theta,c,r))$ with $\mdisc(\theta,c,r)$ given by
        \begin{align*}
        \begin{bNiceArray}{cc}
            C\tp \\ D\tp-ce^{-\ii*\theta} \mati_n
        \end{bNiceArray}
        \begin{bNiceArray}{cc}[cell-space-limits=1pt]
            C & D-ce^{\ii*\theta} \mati_n
        \end{bNiceArray} - 
        \begin{bNiceArray}{cc}[cell-space-limits=1pt]
            \mato & \mato\\
            \mato & r^2 \mati_n
        \end{bNiceArray}\end{align*} is feasible.
    \end{enumerate} 
\end{lemma}
\begin{myproof}
    The equality in 2) defines a hyperplane in $\cF\times \rF$ passing through $((c-r)e^{\ii*\theta},(c-r)^2)$, with normal vector $(2ce^{\ii*\theta},-1)$. Then 2) requires that $\dwshell(G)$ lies in the half space determined by this hyperplane in the direction of the normal vector, which is equivalent to 1) (see \cite[Props. 2 and 3]{zhangPhantomDavisWielandtShell2025a} and \cref{fig:algo_demo}).
    The equivalence between 2) and 3) follows from a direct application of the non-strict version of the generalised KYP lemma \cite{iwasakiGeneralizedKYPLemma2005,liangFeedbackStabilityMixed2025,yangSectoredRealLemma2025}.
\end{myproof}
\begin{remark}
    When $\theta=0\text{ or }\pi$, the SRG condition reduces to the time-domain sector-bounded systems studied in \cite{guptaRobustStabilityAnalysis1996,bridgemanComparativeStudyInput2018}. In particular, 
    when $\theta=0$, $c = r \rightarrow \infty$, the condition characterises the positive real property. 
    In this case, we can divide the matrix inequalities in 2) and 3) by $r$, and take the limit, which gives
    \begin{align*}
        \mdisc(\theta, \infty, \infty) := \begin{bmatrix}
            \mato & - e^{\ii*\theta}C\tp \\
            -e^{-\ii*\theta} C & -(e^{\ii*\theta} D\tp + e^{-\ii*\theta}D)
        \end{bmatrix}.
    \end{align*}
    Through a family of carefully chosen $c, \theta\text{ and }r$, one could approximate more complicated regions.
\end{remark}

Based on \cref{lem:disc-bounded},
\cref{tab:mi-char-ga-ang} lists three types of bounding discs that are instrumental here. 
Note the MIs for gain- and angle-type discs are essentially linear in $r$ and $c$, respectively. The former is linear in $r^2$ while the latter $\mi(X,Y,c;\mdisc(\theta, c,c \sin\,\varphi))$ 
    is equivalent to:  
    \begin{align*}
          \renewcommand{\arraystretch}{1.6}
        \begin{bNiceArray}{c:c}
            -\mati_n & \Block[c]{1-1}{\begin{bmatrix}
                \mato & \, c \cos\varphi \mati_n 
            \end{bmatrix}} \\ \hdottedline
            \Block[c]{1-1}{
            \begin{bmatrix}
                \mato \\ c \cos\varphi \mati_n
            \end{bmatrix}} &
            \Block[c]{1-1}{S(X,Y,c)}
        \end{bNiceArray}
        \preccurlyeq\mato,\quad Y \succcurlyeq\mato,
    \end{align*} with the block $S(X,Y,c)$ being
    \begin{align*}
        \milhs(X, Y, c;\mdisc(\theta,c,c\sin\varphi)) - \begin{bmatrix}
            \mato \\ c\cos\varphi\mati_n
        \end{bmatrix}\begin{bmatrix}
            \mato \\ c\cos\varphi\mati_n
        \end{bmatrix}\ct.
    \end{align*} The block above is linear in $X, Y$ and $c$, and hence the entire constraint is an LMI.
\begin{table}[t]
\centering
\caption{Gain-, angle-, and mixed-type discs.\label{tab:mi-char-ga-ang}}\vspace{-.2cm}
\begin{NiceTabular}{|@{\,}l@{\hspace{.4em}}p{3.6cm}@{\hspace{.6em}}p{3.8cm}@{\,}|}
\hline 
Type & Bounding discs &  Characterisation\\ \hline\hline
Gain & discs centered at origin: $\disc[r](0)$ & $\mi(X,Y,r; \mdisc(0,0,r))$\\
$\theta$-angle & inscribed discs of a cone: $\disc[c\sin\varphi](ce^{\ii*\theta})$&  $\mi(X,Y,c; \mdisc(\theta, c,c\sin\varphi))$\\
Mixed & inscribed discs of the union of a gain-type disc and a cone & 
\begin{minipage}[t]{3.8cm}
    $\lmi(X,Y,\lambda; (1-\lambda)\mdisc(0,0,r)$\\
    $+\lambda\mdisc(\theta, c,c\sin\varphi)), \lambda \in [0,1]$
\end{minipage}    
\\\hline
\end{NiceTabular}\vspace{-.6cm}
\end{table}

\begin{figure}[h]
        \centering\hspace{-1em}
        \subfloat[DW space: bounding halfspaces \label{subfig:algo-3d}]{
            \makebox[.55\columnwidth]{\adjincludegraphics[Clip={.01\width} {.05\height} {.01\width} {.03\height},height=3.6cm]{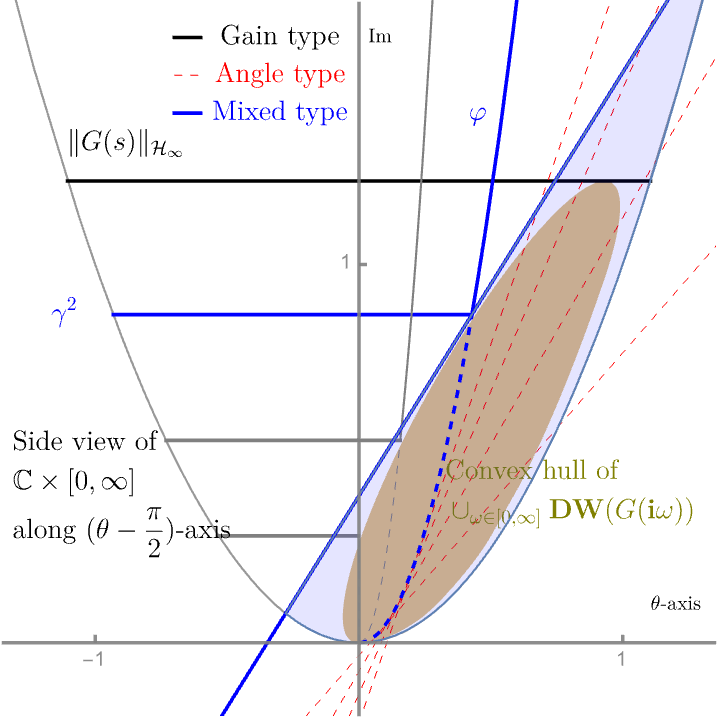}}}
        \subfloat[SRG plane: bounding discs \label{subfig:algo-2d}]{
            \adjincludegraphics[Clip={.04\width} {.1\height} {.005\width} {0},height=3.6cm]{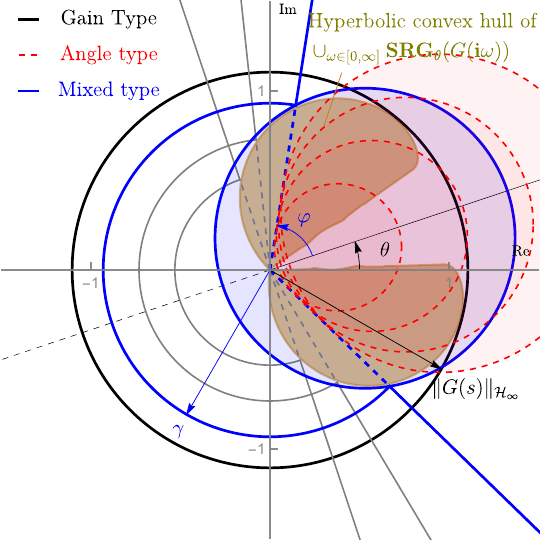}}
            \vspace{-.5cm}
        \caption{Using gain-, angle-, mixed-type halfspaces (discs) to cover the (hyperbolic) convex hull of frequencywise DW shells ($\theta$-SRGs). \label{fig:algo_demo}} \vspace{-.5cm}
\end{figure}

\def\gamarray{\boldsymbol{\Gamma}_\theta}
\def\phiarray{\boldsymbol{\Phi}_\theta}
\def\lamarray{\boldsymbol{\Lambda}}
\def\phiinit{\varphi_{\mathrm{wst}}}
\def\algotheta{\boldsymbol{\theta}}
\begin{algorithm}[!t]
\footnotesize
\caption{Computation of the $\theta$-$\varphi$-$\gamma$ profile of $G$.}
\label{alg:theta-ga-ang}
\KwData{Minimal realisation $(A, B, C, D)$ of $G(s)$, sample sizes $N_\theta$ and $N_\gamma$ along $\theta$-axis and in each $\varphi$-$\gamma$ plane, error bound.}
\KwResult{For each $\theta$, a gain array $\gamarray$ and an angle array $\phiarray$.}

\tcc{{\bfseries Step 1. Compute $\nrmhinf(G)$.}}
\hspace{-.3em}$\gamarray(N_\gamma) \leftarrow$ Minimise $r$ subject to $\mi(X, Y, r; 0,0,r)$; $\phiarray(N_\gamma)\leftarrow 0$\; 
\tcc{{\bfseries Step 2. Search for $\anghinf(G)$.}}
$\algotheta\leftarrow$ uniformly sample $N_\theta$ points from $[-\pi,\pi)$\;
\For{each $\theta \in \algotheta$}{
    \tcc{A half-plane ($\theta$-positive real) test.}
    \eIf{$\lmi(X, Y; \mdisc(\theta,\infty, \infty))$ is feasible}{
        \SetKwBlock{PhiBisect}{$\{\phiarray(0),c_{\theta}\}\leftarrow$ Bisection to minimise $\varphi\in[0,\pi/2]$}{end}
        \PhiBisect{
            At each iteration with a fixed $\varphi$, minimise and maximise $c$ subject to 
            $
                \mi(X, Y, c; \mdisc(\theta,c,c\sin\varphi))
            $.
            If the MI is feasible and the optimal values are $\smal[c],\larg[c]$, then decrease $\varphi$ to $ \arccos(\frac{\smal[c]+\larg[c]}{2\sqrt{\smal[c]\larg[c]}}\cos\,\varphi)$ or to the next bisection point, whichever is smaller; Otherwise, increase $\varphi$. Iterate until error bound is reached and update $\phiarray(0), c_\theta$ to $\varphi, (\smal[c]+\larg[c])/2$, respectively.
        }
    }{\hspace{-4pt}$\phiarray(0) \leftarrow \pi/2$; $c_{\theta}\leftarrow$ 0; 
    \tcc{In fact, $\phiarray(0)>\pi/2$.%
        }}
    \tcc{{\bfseries Step 3. Mixed gain-angle refinement.}}
    \eIf{$\phiarray(0)$ is $\pi/2$}{
        \SetKwBlock{PrBisect}{$\{\gamarray(1),\lamarray(1)\}\leftarrow$ Bisection to minimise $r\in[0,\gamarray(N_\gamma)]$}{end}
        \PrBisect{
            At each iteration with a fixed $r$, maximise $\lambda$ subject to\vspace{-.6em}
            \begin{align*}
            \hspace{-.4cm}\lmi(X,Y,\lambda; (1-\lambda) \mdisc(0,0,r)+\lambda \mdisc(\theta,\infty,\infty)), \lambda \in [0,1].
            \end{align*}
            If the LMI is feasible, then decrease $r$. Otherwise, increase $r$. Iterate until error bound is reached and update $\gamarray(1), \lamarray(1)$ to $r, \lambda$, respectively\;  
        }
            $\phiarray(1)\leftarrow \pi/2$\;
    }{
        $\phiarray(1)\leftarrow \phiarray(0); \quad \gamarray(1)\leftarrow c_{\theta}\sin\phiarray(1)$\;
    }
    
    \eIf{$|\gamarray(N_\gamma)-\gamarray(1)|>$error bound}{
        $\gamarray\leftarrow$ uniformly sample $N_\gamma$ points from $[\gamarray(1),\gamarray(N_\gamma)]$\;
        \For{
            $k = 2:(N_\gamma-1)$
        }{
            $\zeta\leftarrow \frac{1}{2\lamarray(k-1)}\frac{\gamarray(k)}{\gamarray(k-1)}+\left(1-\frac{1}{2\lamarray(k-1)}\right)\frac{\gamarray(k-1)}{\gamarray(k)}$\;
            $\phiinit\leftarrow \arccos\left(\zeta \cos\phiarray(k-1)\right)$\; %
            $\varphi_{\mathrm{bst}}\leftarrow $solve $\frac{1+\sin(\varphi)}{\cos(\varphi)} = \frac{\gamarray(N_\gamma)}{\gamarray(k)}$ for $\varphi \in [0,\pi/2]$\;
            \SetKwBlock{GaPhiBisect}{\hspace{-.8em}$\{\phiarray(k),\lamarray(k)\}\leftarrow$ Bisection to minimise $\varphi\in [\varphi_{\mathrm{bst}},\phiinit]$}{end}
            \GaPhiBisect{
                \hspace{-1em}At each iteration with a fixed $\varphi$, maximise $\lambda$ subject to \vspace{-.6em}
                \begin{align*}
                    &\hspace{-1em}\mathrm{LMI}\left(X,Y,\lambda;(1-\lambda)\mdisc(0,0,\gamarray(k)) \vphantom{\frac{\gamarray(k)}{\cos\varphi}e^{\ii*\theta}}\right.\\
                    &\hspace{0em}   \left. +\lambda\mdisc\left(\theta,\frac{\gamarray(k)}{\cos\varphi},\gamarray(k)\tan\varphi\right)\right), \lambda \in [0,1].
                \end{align*}
                If the LMI is feasible, then decrease $\varphi$. Otherwise, increase $\varphi$. Iterate until error bound is reached and update $\phiarray(k),\lamarray(k)$ to $\varphi, \lambda$, respectively.
            }
        }
    }{
        $\gamarray \leftarrow \{0, \gamarray(N_\gamma), \gamarray(N_\gamma)\};\quad \phiarray\leftarrow \{\pi/2,\pi/2,\pi\}$\;
    }
}
\end{algorithm}

We can leverage the characterisations for three typical discs in \cref{tab:mi-char-ga-ang} to extract the gain, angle, and mixed information of a given $G$, the process of which is detailed in Algo.~\ref{alg:theta-ga-ang}. The algorithm consists of three main steps:
\begin{enumerate}
    \item Find the smallest gain-type discs (minimise $r$) that contains the frequencywise SRG of $G$ (the black disc in \cref{subfig:algo-2d});
    \item For a fixed $\theta$, perform a bisection search over $\theta$-symmetric cones (i.e, over $\varphi$) to identify the smallest one that contains the frequencywise $\theta$-SRG. The containment of $\theta$-SRG by a specific cone (for fixed $\theta, \varphi$) can be reformulated as containment by the family of $\theta$-angle-type discs parameterised by the magnitudes of their centers (the red discs in \cref{subfig:algo-2d} illustrate these $\theta$-angle-type discs, although none of them provides a valid bound in this particular example).
    \item Combine gain-type discs and angle-type discs (using the $S$-procedure) to  further tighten the covering region of the frequencywise $\theta$-SRG (the blue disc in \cref{subfig:algo-2d}).
\end{enumerate}
\begin{remark}
    As explained in the proof of \cref{lem:disc-bounded} and illustrated in \cref{fig:algo_demo}, each disc in the SRG plane corresponds to a halfspace in the DW space. Gain-type and angle-type discs (or halfspaces) are used to determine the gain and $\theta$-angular bounds of the $\theta$-SRG (or DW shell) of $G$. Furthermore, mixed-type discs (halfspaces) exploit the sharp point which lies on the boundary of the union of the gain- and angle-type discs to probe for more precise boundary of the $\theta$-SRG (DW shell).
\end{remark}

The output of Algo.~\ref{alg:theta-ga-ang} is then used to construct a robustness profile with respect to disc, conic, and sectorial uncertainties. The following example illustrates the construction process and the interpretation of this profile.

\begin{example}
    Consider one of the transfer matrices in %
    \cite[Example~2]{zhangPhantomDavisWielandtShell2025a}:
    \begin{align*}
        \renewcommand{\arraystretch}{1.5}
        G(s) = 
        \begin{bNiceArray}{cc}
            \frac{1}{s+3} & 0 \\
            \frac{0.5s}{s+2} & \frac{s}{s+1}
        \end{bNiceArray}=
        \ssreal(
            {\left[\begin{smallmatrix} -3 & 0 & 0\\0 & -2 & 0\\ 0 & 0 & -1\end{smallmatrix}\right]}
        )(
            {\left[\begin{smallmatrix}1 & 0 \\ 1 & 0 \\ 0 & 1 \end{smallmatrix}\right]}
        )(
            {\left[\begin{smallmatrix}1 & 0 & 0 \\ 0 & -1 & -1 \end{smallmatrix}\right]}
        )(
            {\left[\begin{smallmatrix} 0 & 0 \\ 0.5 & 1 \end{smallmatrix}\right]}
        ).
    \end{align*}
    
    \emph{Construction of the robustness profile}. By implementing Algo.~\ref{alg:theta-ga-ang} on $G(s)$, we obtain, for each $\theta \in [-\pi,\pi)$, a gain array $\gamarray$ and a $\theta$-angle array $\phiarray$. We then construct the profile shown in \cref{fig:r-profile} as follows.
    The ambient space can be viewed as $\cF \times [0,\infty)$. Each $\theta$ determines a ray in $\cF$, to which we attach a vertical $(\varphi,\gamma)$-orthant (see the dark gray orthant in \cref{subfig:the-phi-ga}). Within this $\theta$-orthant, we plot $\gamarray$ against $\phiarray$. Sweeping over all $\theta$ then yields the $\theta$–$\varphi$–$\gamma$ profile of $G(s)$ (the yellow surface in \cref{subfig:the-phi-ga}). Since the $\hinf$-$\theta$-angle takes value in $[0,\pi]$, the profile is only well defined over $\disc[\pi]$ (the disc enclosed within the black region).
    \begin{figure}[t]
        \centering\hspace{-1em}
        \subfloat[$\theta$-$\varphi$-$\gamma$ profile \label{subfig:the-phi-ga}]{
            \adjincludegraphics[Clip={.1\width} {.15\height} {.15\width} {.05\height},width=.5\columnwidth]{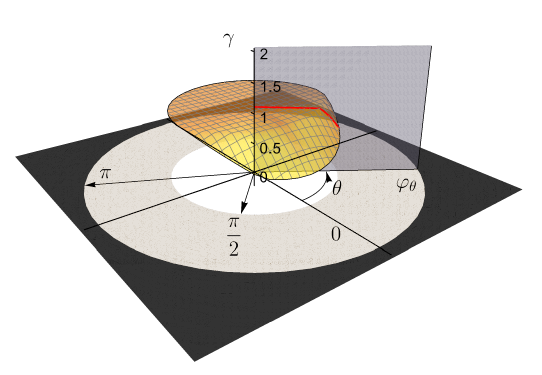}}
        \subfloat[$\theta$-$(\pi-\varphi)$-$1/\gamma$ profile \label{subfig:the-phi-ga-inv}]{
            \adjincludegraphics[Clip={.1\width} {.15\height} {.15\width} {.05\height},width=.5\columnwidth]{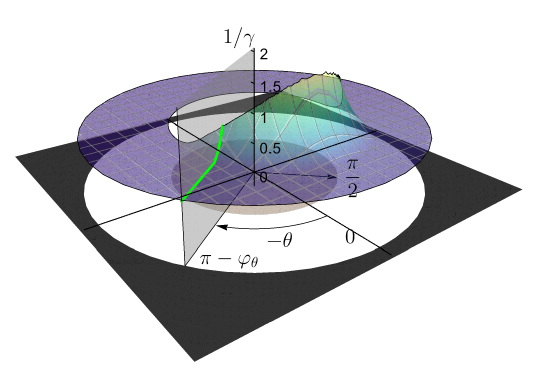}}
        \caption{Gain-angle robustness profile of a transfer matrix. \label{fig:r-profile}} \vspace{-2em}
    \end{figure}

    \emph{Interpretation}. Each point $(\varphi_0 e^{\ii*\theta_0},\gamma_0)$ on the profile shown in \cref{subfig:the-phi-ga} characterises the smallest region of the form $\disc[\gamma_0]\cup\cone*[\theta_0-\varphi,\theta_0+\varphi]$ with fixed radius $\gamma_0$ and varying half angular spread $\varphi$ (or, equivalently, $\disc[\gamma]\cup\cone*[\theta_0-\varphi_0,\theta_0+\varphi_0]$ with fixed $\varphi_0$ and varying radius $\gamma$), that covers the hyperbolic convex hull of $\cup_{\omega\in[0,\infty]} \srg[\theta_0](G(\ii*\omega))$. The corresponding optimal half angular spread (or radius) is given by $\varphi_0$ (or $\gamma_0$).
    As mentioned, we can also plot the complementary $\theta$-$(\pi-\varphi)$-$1/\gamma$ profile. Then each point $(\psi_0 e^{\ii*\vartheta_0}, \eta_0)$ on this profile guarantees that $G(s)$ is robustly stable with respect to the uncertainty set $\sesys_{\eta_0-\varepsilon_{\eta}}[-\vartheta_0-\psi_0-\varepsilon_{\psi},-\vartheta_0+\psi_0+\varepsilon_{\psi}]$ for arbitrarily small $\varepsilon_{\eta},\varepsilon_{\psi}>0$. So the values of $\theta$ for which the complementary profile expands towards the upper-right corner (within the local orthant) indicate directions along which the system exhibits stronger robustness against sectorial uncertainties.
\end{example}
\section{CONCLUSIONS}
We have studied the robustness of MIMO LTI feedback interconnections with respect to uncertainties induced by geometric regions in the SRG plane, in particular discs, cones, and sectors. The ($\theta$)-SRG is used as a graphical system representation, and we propose two ways of defining uncertainties based on a prescribed region in the SRG plane, one with a prior specification on $\theta$ and one without. The robust stability problem is dissected into a matrix robust nonsingularity problem and a system-level problem. 
At the matrix level, we unveil the fundamental requirements on the inducing regions, under both uncertainty definitions, such that SRG-based separation provides a necessary and sufficient condition for matrix robust nonsingularity. This is established through the lens of DW shell and its connection with ($\theta$)-SRGs.
By a frequencywise application of the matrix results, we derive system-level robust stability conditions with respect to a broad class of graphically induced uncertainties. 
Finally, for uncertainties induced by simple regions such as discs, cones, and sectors, we develop algorithms, no more complex than a bisection search with an LMI solved at each iteration, to construct a 3-D robustness profile of a system. This profile provides an intuitive and informative visualisation from which the robustness against sectorial uncertainties can be readily assessed.

\appendix
The proofs are largely geometric, which will be easier to perceive if readers familiarise themselves with the projection relationship between DW shells and $\theta$-SRGs, as well as the matrix tomography idea discussed in \cite[Sec.~VI]{zhangPhantomDavisWielandtShell2025a}. 
\begin{figure}[h]
    \centering
    \includegraphics[height=4cm]{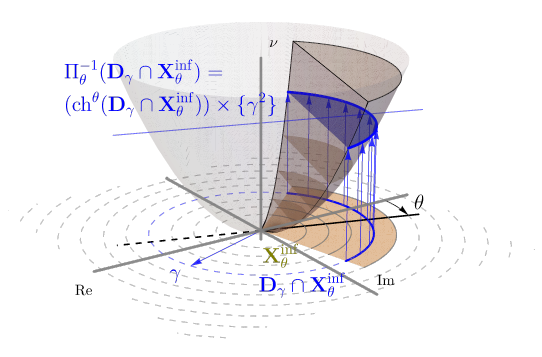} \vspace{-.6cm}
    \caption{Dissect the $\theta$-SRG region with concentric discs centered at the origin, and then project back to the DW space. \label{fig:pf-disc-dissect}}\vspace{-2em}
\end{figure}

\subsection{Proof of \Cref{prop:dw-union-general}\label{pf:dw-union-general}} 
We first show the \emph{formula for the $\theta$-uncertainty set}.
Let $\svmin(\Delta)$ and $\svmax(\Delta)$ denote the smallest and largest singular values of $\Delta$, respectively.
    For an arbitrary $\Delta\in\symset_\theta$, by definition $\srg[\theta](\Delta)\subseteq \infreg[\theta]$. Then it follows that the height interval $[\svmin(\Delta)^2, \svmax(\Delta)^2]$ of $\dwshell(\Delta)$ is always contained in $[\gmin(\infreg[\theta])^2, \gmax(\infreg[\theta])^2]$.
We shall investigate the cross sections of $\dwshell(\Delta)$ by horizontal hyperplanes (layer by layer) and relate them, via $\proj[\theta]$, to intersections between discs and $\srg[\theta](\Delta)$  (circle by circle).
\begin{align*}
    &\dwshell(\Delta) = \cup_{\gamma}\, \dwshell(\Delta) \cap \plane[\gamma^2] \\
    &\subseteq \cup_{\gamma}\, \proj[\theta]\inv\, (\proj[\theta] (\dwshell(\Delta) \cap \plane[\gamma^2])) \\
    &= \cup_{\gamma}\, \proj[\theta]\inv\, (\disc[\gamma]\cap \srg[\theta](\Delta)) \subseteq \cup_{\gamma}\, \proj[\theta]\inv\, (\disc[\gamma]\cap \infreg[\theta])  \\
    &\subseteq \cup_\gamma\, \left((\chordhull((\disc[\gamma]\cap\infreg[\theta])))\times \{\gamma^2\}\right). \numberthis \label{eq:dw_in_rhs}
\end{align*}
Note that $\proj[\theta]\inv\,(\cdot)$ denotes the preimage of a $\theta$-symmetric set under the mapping $\proj[\theta]$ with domain $\epi(\parab[1])$. The interval over which $\gamma$ ranges in the unions above can be taken as $[\gmin(\infreg[\theta]), \gmax(\infreg[\theta])]$ due to the aforementioned singular value containment. Also note that the last ``$\subseteq$'' becomes an equality when the set intersected with $\disc[\gamma]$ is $\theta$-symmetric (which $\infreg[\theta]$ clearly is). This shows the left-hand side (LHS) is contained in the right-hand side (RHS).

Conversely, when the size of matrices is at least 2, for any given point $(z,\gamma^2)$ on the RHS, there must exist, by definition, $(s,\gamma^2) \in \parab[1]$ where $s \in \disc[\gamma]\cap \infreg[\theta]$ and $z$ lies on the line segment connecting $s$ and $s_\theta$. 
    Let $\Delta = \blkdiag{s\mati_{m}, s_\theta \mati_{n-m}}$ where $m$ can be any integer from $1$ to $n-1$, then it is straightforward to check that $\srg[\theta](\Delta) \subseteq \infreg[\theta]$ and 
    hence $\Delta \in \symset_\theta$ while $(z,\gamma^2) \in \dwshell(\Delta)$. 
    The RHS is therefore contained in the LHS. The equality for the $\theta$-uncertainty set is then established.

The \emph{equality for the general uncertainty} $\symset$ can be shown following similar arguments. For any $\Delta \in \symset$, there exists some $\vartheta$ such that $\srg[\vartheta](\Delta)\subseteq \symreg$. Then, by setting $\theta = \vartheta$ in \cref{eq:dw_in_rhs}, we have
\begin{align*}
    \dwshell(\Delta) 
    &\subseteq  \cup_\gamma\, \left((\chordhull[\vartheta]((\disc[\gamma]\cap\infreg[\vartheta])))\times \{\gamma^2\}\right)\\
    &\subseteq \cup_\gamma \, \left((\cup_{\phi}\,\chordhull[\phi]((\disc[\gamma]\cap\symreg)))\times \{\gamma^2\}\right) \\
    &= \cup_\gamma\, (\convhull((\disc[\gamma] \cap \symreg)) \times \{\gamma^2\}).
\end{align*}
The reverse containment (for matrices of size at least 2) follows analogously to the $\theta$-uncertainty case by constructing a scalar multiple of a unitary matrix.

In the scalar case, $\dwshell(\Delta)= \{(\Delta,|\Delta|^2)\}$ is a point confined to $\parab[1]$, and $\srg[\vartheta](\Delta) = \{\Delta,\Delta_\vartheta\}$ is a $\vartheta$-conjugate pair, which collapses to $\Delta$ when $\vartheta = \angle\Delta \mod\pi$. It follows that $\symset_\theta = \{\Delta \in \cF: \Delta \in \infreg[\theta]\} = \infreg[\theta]$ and $\symset = \{\Delta \in \cF: \Delta \in \symreg\} = \symreg$.  
Then it is fairly straightforward to establish that the unions on the LHS are the intersections between their corresponding RHS and $\parab[1]$.

\subsection{Proof of \Cref{thm:mat-theta}\label{pf:theta-uncert}}
Under the assumption that $\symreg$ is symmetric about the $\theta$-axis (which means $\symreg = \infreg[\theta] = \supreg[\theta]$ and $\symset_\theta = \infreg*[\theta] = \supreg*[\theta]$), it follows from \cref{prop:dw-union-general} that 
\begin{align*}
        \proj[\theta]\inv \symreg = \bigcup_{\gmin(\symreg)\leqsl \gamma \leqsl \gmax(\symreg)} 
        (\chordhull((\disc[\gamma] \cap \symreg)) \times \{\gamma^2\} \vphantom{\chordhull(\{\disc[\gamma] \cap \symreg_\theta\})}) .
\end{align*}
Then by \cref{corol:mrn-dw,prop:dw-union-general}, MRN holds for $G$ with respect to $\symset_\theta$ if and only if $\invdwshell(-G)\cap \proj[\theta]\inv \symreg  = \emptyset$. 
Unlike \cref{prop:ga-co-sec-char,prop:dw-union-general}, where our goal is to precisely characterise the union of DW shells, we do not distinguish a scalar case here (or in Appendix~\ref{pf:general-uncert}), since in this case, $\invdwshell(-G)$ always resides in $\parab[1]$ and its separation from $\proj[\theta]\inv\,\symreg$ (or $\proj[\theta]\inv\,\circhull(\symreg)$) is equivalent to separation from $\dwshell(\symset_\theta)$ (or $\dwshell(\symset)$).

Apply $\proj[\theta]$ to the two sets in the separation condition and it follows that the 3-D separation condition holds if $\invsrg[\theta](-G)\cap \symreg=\emptyset$. 
Conversely, suppose that there exists some $z \in \invsrg[\theta](-G)\cap \symreg$.
    Since $z\in \invsrg[\theta](-G)$, there must exist some $\ztld$ on the line segment $\lineseg(z,z_\theta)$ connecting $z$ and $z_\theta$ such that $(\ztld, |z|^2) \in \invdwshell(-G)$.
    Meanwhile, due to the symmetry of $\symreg$, $z\in \symreg$ implies that $z_\theta\in \symreg$. Therefore, $\lineseg(z,z_\theta)\subseteq \chordhull((\disc[|z|] \cap \symreg))$ and we have
    \begin{align*}
        (\ztld, |z|^2) &\in \lineseg(z,z_\theta)\times \{|z|^2\}\\
        &\subseteq  \chordhull((\disc[|z|] \cap \symreg)) \times \{|z|^2\} 
        \subseteq \proj[\theta]\inv\,\symreg.
    \end{align*}
    This shows that $(\ztld, |z|^2) \in \invdwshell(-G)\cap \proj[\theta]\inv\,\symreg$. Hence, $\invdwshell(-G)\cap \proj[\theta]\inv\,\symreg$ also implies $\invsrg[\theta](-G)\cap \symreg = \emptyset$.
    The proof is now complete.
\subsection{Proof of \Cref{thm:mat-general} \label{pf:general-uncert}}
\textbf{Sufficiency}:
    Note that, since $\circhull(\symreg)$ is $\theta$-symmetric,
    \begin{align*}
        \proj[\theta]\inv \circhull(\symreg) = 
        \bigcup_{\gmin(\symreg)\leqsl\gamma\leqsl\gmax(\symreg)} \chordhull((\disc[\gamma]\cap \circhull(\symreg))) \times \{\gamma^2\}.
    \end{align*}
    Furthermore, it holds that $\convhull((\disc[\gamma]\cap \symreg)) \subseteq \chordhull((\disc[\gamma]\cap \circhull(\symreg)))$ for all $\gamma$. Thus, by \cref{prop:dw-union-general}, we have:
    $\dwshell(\symset)\subseteq \proj[\theta]\inv \circhull(\symreg)$, indicating $\proj[\theta]\dwshell(\symset) \subseteq \circhull(\symreg)$.
    It follows that
    \begin{align*}
        &\invsrg[\theta](-G)\cap \circhull(\symreg) = \emptyset \\
        &\hspace{1cm}\Rightarrow
        \invsrg[\theta](-G)\cap \proj[\theta]\dwshell(\symset) = \emptyset \\
        &\hspace{2cm}\Rightarrow
        \invdwshell(-G) \cap \dwshell(\symset) = \emptyset, 
    \end{align*} which further implies MRN holds for $G$ with respect to $\symset$ by \cref{corol:mrn-dw}.

\textbf{Necessity} when $\symreg$ is $\theta$-symmetric and $\theta$-circularly connected: 
Recalling \cref{obs:circhull}, $\symreg = \circhull(\symreg)$ and thus 
    \begin{align*}
        \chordhull((\disc[\gamma]\cap\circhull(\symreg)))
        = \chordhull((\disc[\gamma]\cap\symreg))
        = \convhull((\disc[\gamma]\cap \symreg)) 
    \end{align*} for all $\gamma \in [\gmin(\symreg),\gmax(\symreg)]$.
This together with \cref{prop:dw-union-general} shows that $\dwshell(\symset) = \dwshell(\symset_\theta)$. Therefore, the following statements are all equivalent:
\begin{enumerate}
    \item MRN holds for $G$ with respect to $\symset$;
    \item $\invdwshell(-G)\cap\dwshell(\symset) = \emptyset$;
    \item $\invdwshell(-G)\cap\dwshell(\symset_\theta) = \emptyset$;
    \item $\invsrg[\theta](-G)\cap\symreg = \emptyset$; 
    \item $\invsrg[\theta](-G)\cap \circhull(\symreg) = \emptyset$.
\end{enumerate}
where 3) $\Leftrightarrow$ 4) was already shown in \cref{thm:mat-theta}. 

\textbf{Lack of universal necessity} when any of the $\theta$-symmetry and $\theta$-circular connectedness fails to hold: We show that whenever one of these two properties is missing, then a $G$ can be constructed such that $\invsrg[\theta](-G) \cap \circhull(\symreg)$ is nonempty but MRN still holds for $G$ with respect to $\symset$:
\begin{itemize}
    \item \emph{Not $\theta$-symmetric}: In this case, there must exist some $z\in \symreg$ such that $z_\theta\notin \symreg$ (then such $z$ must not be zero). However, $z_\theta \in \circhull(\symreg)$ by its definition.
    Let $G = - \frac{1}{z_\theta} \mati_n$. Note that $\invdwshell(-G)$ is a singleton $\{(z_\theta,|z_\theta|^2)\}$, and this point is not inside $\convhull((\disc[|z_\theta|^2]\cap\symreg))\times \{|z_\theta|^2\}$. Therefore, it follows from \cref{corol:mrn-dw} and \cref{prop:dw-union-general} that MRN holds for $G$ with respect to $\symset$. However, $\invsrg[\theta](-G) = \{z_\theta\} \subseteq \circhull(\symreg)$.
    \item \emph{Not $\theta$-circularly connected}: There must exist some $\gamma\in [\gmin(\symreg),\gmax(\symreg)]$ such that $\disc[\gamma]\cap \symreg$ has a disconnected intersection with at least one of the two half-planes determined by the $\theta$-axis. 
    Choose $z$ from a gap in this disconnected intersection (i.e., $|z| = \gamma$ and $|\angle z| \in [\angmin[\theta](\symreg;\gamma),\angmax[\theta](\symreg;\gamma)]$, but $z \notin \symreg$).
    Then $G = -\frac{1}{z} \mati_n$ serves as an example which violates the SRG condition but still satisfies MRN with respect to $\symset$. 
\end{itemize}

\subsection{Proof of \Cref{prop:sym-set-char} \label{pf:sym-set-char}}
It follows directly from definitions that $\symset_\theta \subseteq \symset$ in general. 
By \cref{prop:dw-union-general}, any $\Delta$ from $\symset$ will satisfy:
    $
        \dwshell(\Delta) \subseteq \dwshell(\symset).
    $
    When $\symreg$ is both $\theta$-symmetric and $\theta$-circularly connected, 
        we have $\dwshell(\symset) = \dwshell(\symset_\theta)$ (see the necessity part in Appendix~\ref{pf:general-uncert}) and hence $\dwshell(\Delta)\subseteq \dwshell(\symset_\theta)$.
    Then apply $\proj[\theta]$ to both sides and we have $\srg[\theta](\Delta) \subseteq \symreg$, which shows that $\Delta \in \symset_\theta$. 

If  $\theta$-symmetry fails, 
    we choose $\Delta = z \mati_n$ where $z \in \symreg \setminus \infreg[\theta]$.
If $\theta$-circular connectedness fails,
    we choose $\Delta = \blkdiag{z_1\mati_m, z_2 \mati_{n-m}}$ where $m$ can be any integer from $1$ to $n-1$ and $z_1, z_2$ are taken from two disconnected arcs of $\disc[|z|]\cap\symreg$ that lie on the same side of the $\theta$-axis.
A direct verification then shows that $\Delta$ is in $\symset$ but not in $\symset_\theta$.

\bibliographystyle{IEEEtran}

\end{document}